\documentclass[aoas, preprint]{imsartarxiv}
\RequirePackage[OT1]{fontenc}
\usepackage{subcaption, adjustbox} 
\usepackage{amsthm, amsmath, amsfonts, enumerate, mathtools, natbib}
\RequirePackage[colorlinks, citecolor=blue, urlcolor=blue]{hyperref}
\usepackage{verbatim} % for comment block
\usepackage{imsartarxiv}
\usepackage[dvipsnames]{xcolor}

\colorlet{highlight}{Black}

\usepackage{cleveref}
\crefformat{equation}{(#2#1#3)} % \cref{eq1} 
\crefrangeformat{equation}{(#3#1#4) to~(#5#2#6)} % \crefrange{eq1}{eq5}
\crefformat{enumi}{step~#2#1#3} % \cref{step1} 
\crefrangeformat{enumi}{steps~#3#1#4~--~#5#2#6}
\crefformat{figure}{Figure~#2#1#3}

\startlocaldefs
\numberwithin{equation}{section}
\theoremstyle{plain}
\newtheorem{lemma}{Lemma}

\newtheorem{theorem}{Theorem}

\long\def\@makecaption#1#2{%
  \vskip\abovecaptionskip
  \footnotesize
  \sbox\@tempboxa{\itshape\textsc{#1}. #2}%
  \ifdim \wd\@tempboxa >\hsize
    \itshape\textsc{#1}. #2\par
  \else
    \global \@minipagefalse
    \hb@xt@\hsize{\hfil\box\@tempboxa\hfil}%
  \fi
  \vskip\belowcaptionskip}

\long\def\@makecaption#1#2{%
  \vskip\abovecaptionskip
  \figurecaption@size
  \sbox\@tempboxa{{\figurename@size #1}\figurename@skip #2}%
  \ifdim \wd\@tempboxa >\hsize
    {\figurename@size #1}\figurename@skip #2\par
  \else
    \global \@minipagefalse
    \hb@xt@\hsize{\hfil\box\@tempboxa\hfil}%
  \fi
  \vskip\belowcaptionskip}

\setattribute{figurecaption}{size}{\footnotesize\itshape}
\setattribute{figurename}   {size}{\scshape}
\setattribute{figurename}   {skip}{.~~}

\def\@floatboxreset{%
        \reset@font
        \@setminipage
        \singlespacing
        \footnotesize
        \centering
}

\endlocaldefs

\newcommand{ \Prob  } {\mathbb{P} }
\newcommand{ \Exp  } {\mathbb{E} }

\newcommand{\beginsupplement}{%
        \setcounter{table}{0}
        \renewcommand{\thetable}{S\arabic{table}}%
        \setcounter{figure}{0}
        \renewcommand{\thefigure}{S\arabic{figure}}%
         \setcounter{section}{0}
        \renewcommand{\thesection}{S\arabic{section}}%
}

\bgroup
\def\arraystretch{1.5} % table vertical spacing; 1 is the default
\newcommand{\bftab}{\fontseries{b}\selectfont} % bf without changing the table width

\begin{document}

%!TEX root = CovTest-sPCA.tex

\begin{frontmatter}

\title{Testing High Dimensional Covariance Matrices, with Application to Detecting Schizophrenia Risk Genes}
\runtitle{Testing High Dimensional Covariance Matrices}

%% authors
\begin{aug}
\author{\fnms{Lingxue} \snm{Zhu}\thanksref{t1}\ead[label=e1]{lzhu@cmu.edu}},
\author{\fnms{Jing} \snm{Lei}\thanksref{t1}\ead[label=e2]{jinglei@andrew.cmu.edu}},
\author{\fnms{Bernie} \snm{Devlin}\thanksref{t2}\ead[label=e3]{devlinbj@upmc.edu}}
\and
\author{\fnms{Kathryn} \snm{Roeder}\corref{}\thanksref{t1}\ead[label=e4]{roeder@andrew.cmu.edu}}

\affiliation{Carnegie Mellon University\thanksmark{t1}
and
University of Pittsburgh\thanksmark{t2}}

\address{Department of Statistics\\
Carnegie Mellon University\\
5000 Forbes Avenue\\
Pittsburgh, Pennsylvania 15213\\
USA\\
\printead{e1}\\
\phantom{E-mail:\ }\printead*{e2}\\
\phantom{E-mail:\ }\printead*{e4}}

\address{Department of Psychiatry and Human Genetics\\
University of Pittsburgh School of Medicine\\
3811 O'Hara Street\\
Pittsburgh, Pennsylvania 15213\\
USA\\
\printead{e3}}

\runauthor{L. Zhu, J. Lei, B. Devlin and K. Roeder}
\end{aug}

% abstract
\begin{abstract}
Scientists routinely compare gene expression levels in cases versus controls in part to determine genes associated with a disease.  Similarly, detecting case-control differences in co-expression among genes can be critical to understanding complex human diseases; however statistical methods have been limited by the high dimensional nature of this problem.  
In this paper, we construct a sparse-Leading-Eigenvalue-Driven (sLED) test for comparing two high-dimensional covariance matrices.
By focusing on the spectrum of the differential matrix, sLED provides a novel perspective that accommodates
what we assume to be common, namely sparse and weak signals in gene expression data, and it is closely related with Sparse Principal Component Analysis. We prove that sLED achieves full power asymptotically under mild assumptions, and simulation studies verify that it outperforms other existing procedures under many biologically plausible scenarios. Applying sLED to the largest gene-expression dataset obtained from post-mortem brain tissue from Schizophrenia patients and controls, we provide a novel list of genes implicated in Schizophrenia and reveal intriguing patterns in gene co-expression change for Schizophrenia subjects. 
We also illustrate that sLED can be generalized to compare other gene-gene ``relationship'' matrices that are of practical interest, such as the weighted adjacency matrices.

\end{abstract}

% keywords
\begin{keyword}
\kwd{Permutation test}
\kwd{high-dimensional data}
\kwd{covariance matrix}
\kwd{sparse principal component analysis.}
\end{keyword}

\end{frontmatter}

%!TEX root = CovTest-sPCA.tex

%!TEX root = ./CovTest-sPCA.tex

\section{Introduction}
\label{sec:introduction}

High throughput technologies provide the capacity for measuring potentially interesting genetic features on the scale of tens of thousands. With the goal of understanding various complex human diseases, a widely used technique is gene differential expression analysis, which focuses on the marginal effect of each variant. Converging evidence has also revealed the importance of co-expression among genes, but analytical techniques are still underdeveloped. Improved methods in this domain will enhance our understanding of how complex disease affects the patterns of gene expression, shedding light on both the development of disease and its pathological consequences.  

Schizophrenia (SCZ), a severe mental disorder with 0.7\% lifetime risk \citep{mcgrath:e2008a}, is one of the complex human traits that has been known for decades to be highly heritable but whose genetic etiology and pathological consequences remain unclear. What has been repeatedly confirmed is that a large proportion of SCZ liability traces to polygenetic variation involving many hundreds of genes together, with each variant exerting a small impact \citep{purcell:n2014a, international-schizophrenia-consortium:n2009a}. Despite the large expected number, only a small fraction of risk loci have been conclusively identified \citep{schizophrenia-working-group-of-the-psychiatric-genomics-consortium:n2014a}. This failure is mainly due to the limited signal strength of individual variants and under-powered mean-based association studies. 
Still, several biological processes, including synaptic mechanisms and glutamatergic neurotransmission, have been reported to be implicated in the risk for SCZ \citep{CMC2016}. 
The observation that each genetic variant contributes only moderately to risk, and that each affected individual carries many risk variants, suggests that SCZ develops as a consequence of subtle alterations of both gene expression and co-expression, which requires development of statistical methods to describe the subtle, wide-spread co-expression differences.

Pioneering efforts have started in this direction. Very recently, the CommonMind Consortium (CMC) completed a large-scale RNA sequencing on dorsolateral prefrontal cortex from 279 control and 258 SCZ subjects, forming the largest brain gene expression data set on SCZ \citep{CMC2016}. 
Analyses of these data by the CommonMind Consortium suggest that many genes show altered expression between case and control subjects, although the mean differences are small. By combining gene expression and co-expression patterns with results from genetic association studies, it appears that genetic association signals tend to cluster in certain sets of tightly co-expressed genes, so called co-expression modules \citep{zhang2005general}. Still, the study of how gene co-expression patterns change from controls to SCZ subjects remains incomplete. Here, we address this problem using a hypothesis test that compares the gene-gene covariance matrices between control and SCZ samples, with integrated variable selection.

The problem of two-sample test for covariance matrices has been thoroughly studied in traditional multivariate analysis \citep{anderson1958introduction}, but becomes nontrivial once we enter the high-dimensional regime. Most of the previous high-dimensional covariance testing methods are motivated by either the $L_2$-type distance between matrices where all entries are considered \citep{schott2007test, li2012two}, or the $L_\infty$-type distance where only the largest deviation is utilized \citep{cai2013two, chang2015bootstrap}. These two strategies are designed for two extreme situations, respectively: when almost all genes exhibit some difference in co-expression patterns, or when there is one ``leading'' pair of genes whose co-expression pattern has an extraordinary deviation in two populations. 
However, the mechanism of SCZ is most likely to lie somewhere in between, where the difference may occur among hundreds of genes (compared to a total of $\approx\,$20,000 human genes), yet each deviation remains small.
\textcolor{highlight}{
Some other existing approaches include using the trace of the covariance matrices \citep{srivastava2010testing}, using random matrix projections \citep{wu2015tests}, and using energy statistics to measure the distance between two populations \citep{szekely2013energy}. But none of these methods are designed for the scenario in which the signals are both sparse and weak. 
}

In this paper, we propose a sparse-Leading-Eigenvalue-Driven (sLED) test. It provides a novel perspective for matrix comparisons by evaluating the spectrum of the differential matrix, defined as the difference between two covariance matrices. This provides greater power and insight for many biologically plausible models, including the situation where only a small cluster of genes has abnormalities in SCZ subjects, so that the differential matrix is supported on a small sub-block. 
The test statistic of sLED links naturally to the fruitful results in Sparse Principle Component Analysis (SPCA), which is widely used for unsupervised dimension reduction in the high-dimensional regime. Both theoretical and simulation results verify that sLED has superior power under sparse and weak signals.
In addition, sLED can be generalized to comparisons between other gene-gene ``relationship" matrices, including the weighted adjacency matrices that are commonly used in gene clustering studies \citep{zhang2005general}. Applying sLED to the CMC data sheds light on novel SCZ risk genes, and reveals intriguing patterns that are previously missed by the mean-based differential expression analysis.

For the rest of this paper, we motivate and propose sLED for testing two-sample covariance matrices in \Cref{sec:methods}. We provide two algorithms to compute the test statistic, and establish theoretical guarantees on the asymptotic consistency.   In \Cref{sec:simulations}, we conduct simulation studies and show that sLED has superior power to other existing two-sample covariance tests under many scenarios.  In \Cref{sec:data}, we apply sLED to the CMC data. We detect a list of genes implicated in SCZ and reveal interesting patterns of gene co-expression changes. We also illustrate that sLED can be generalized to comparing weighted adjacency matrices. \Cref{sec:discussion} concludes the paper and discusses the potential of applying sLED to other datasets. All proofs are included in the Supplement.  
An implementation of sLED is provided at \url{https://github.com/lingxuez/sLED}.

%!TEX root = CovTest-sPCA.tex

\section{Methods}
\label{sec:methods}

%%% Settings %%%
\subsection{Background} 
%Suppose RNA-sequencing is conducted on $p$ genes for $n$ case samples $X_1, \cdots, X_n \stackrel{i.i.d.}{\sim} ({\mathbf 0}_p, \Sigma_1)$ and $m$ control samples $Y_1, \cdots, Y_m \stackrel{i.i.d.}{\sim} ({\mathbf 0}_p, \Sigma_2)$, where $X_i, Y_j \in \mathbb{R}^p$ independently come from two populations with potentially different covariance matrices. 
Suppose $X_1, \cdots, X_n \stackrel{i.i.d.}{\sim} ({\mathbf 0}_p, \Sigma_1)$ and $Y_1, \cdots, Y_m \stackrel{i.i.d.}{\sim} ({\mathbf 0}_p, \Sigma_2)$ are independent $p$-dimensional random variables coming from two populations with potentially different covariance structures. 
Without loss of generality, both expectations are assumed to be zero, and let $D = \Sigma_2 - \Sigma_1$ be the differential matrix. The goal is to test 
\begin{equation}
H_0: \, D = 0 \  \textrm{versus }\   H_1: \, D \neq 0\,. 
\label{eq:test-setup}
\end{equation}

This two-sample covariance testing problem has been well studied in the traditional ``large $n$, small $p$" setting, where the likelihood ratio test (LRT) is commonly used. However, testing covariance matrices under the high-dimensional regime is a nontrivial problem. In particular, LRT is no longer well defined when $p > \min \{n, m\}$. Even if $p \leq \min\{n, m\}$, LRT has been shown to perform poorly when $p / \min\{n, m\} \to c \in (0, 1)$  \citep{bai2009corrections}. 

Researchers have approached this problem in different ways. Here, we give detailed review on two of the main strategies to motivate our test. The first one starts from rewriting \cref{eq:test-setup} as
\begin{equation}
H_0: \, || D ||_F^2 = 0 \  \textrm{versus }\   H_1: \, || D ||_F^2 \neq 0\,,
\label{eq:test-f2}
\end{equation}
where $||D||_F$ is the Frobenius norm of $D$. 
\textcolor{highlight}{
This strategy includes a test statistic based on an estimator of $||D||_F^2$ under normality assumptions \citep{schott2007test}, as well as a test under more general settings using a linear combination of three U-statistics, which is also motivated by $||D||_F^2$ \citep{li2012two}. These $L_2$-norm driven tests target a dense alternative, but usually suffer from loss of power when $D$ has only a small number of non-zero entries. 
}
%A closely related test for multivariate normal populations, proposed by \cite{srivastava2010testing}, is based on the quantity $\textrm{tr}(\Sigma_1^2) / (\textrm{tr}(\Sigma_1))^2 - \textrm{tr}(\Sigma_2^2) / (\textrm{tr}(\Sigma_2))^2$. This test has also been observed to have limited power when $D$ is sparse \citep{cai2013two}.  
On the other hand, \cite{cai2013two} consider the sparse alternative, and rewrite \cref{eq:test-setup} as
\begin{equation}
H_0: \, ||D||_\infty = 0 \  \textrm{versus }\   H_1: \,||D||_\infty \neq 0\,,
\label{eq:test-max}
\end{equation}
where $||D||_\infty = \max_{i, j} |D_{ij}|$.
Then the test statistic is constructed using a normalized estimator of $||D||_\infty$. Later, \cite{chang2015bootstrap} propose a bootstrap procedure using the same test statistic but under weaker assumptions, and \cite{cai2015inference} extend the idea to comparing two-sample correlation matrices. 
These $L_\infty$-norm based tests have been shown to enjoy superior power when the single-entry signal is strong, in the sense that $\max_{i,j} | D_{ij} |$ is of order $\sqrt{ \log p /  \min \{n, m\} }$ or larger. 

In this paper, we focus on the unexplored but practically interesting regime where the signal is both sparse and weak, meaning that the difference may occur at only a small set of entries, while the magnitude tends to be small. We propose another perspective to construct the test statistic by looking at the singular value of $D$, which is especially suitable for this purpose. To illustrate the idea, consider a toy example where 
\begin{equation}
 D_{ij} = \begin{cases}
\delta, & 1 \leq i, j \leq s\\
0, & \textrm{otherwise}
\end{cases} 
\label{eq:toy-ex}
\end{equation}
for some $\delta > 0$ and integer $s \ll p$. In other words, $\Sigma_2$ and $\Sigma_1$ are only different by $\delta$ in an $s \times s$ sub-block. In this case, the $L_2$-type tests are sub-optimal because they include errors from all entries; so are the $L_{\infty}$-type tests because they only utilize one single entry $\delta$. On the other hand, the largest singular value of $D$ is $s \delta$, which extracts stronger signals with much less noise and therefore has the potential to gain more power. 

More formally, we rewrite the testing problem \cref{eq:test-setup} to be
\begin{equation}
\begin{split}
& H_0: \, \sigma_1(D) = 0  \  \textrm{versus }\    H_1: \, \sigma_1(D) \neq 0\,,
\end{split} 
\label{eq:test-D}
\end{equation}
where $\sigma_1(\cdot)$ denotes the largest singular value. Compared to \cref{eq:test-f2} and \cref{eq:test-max}, \cref{eq:test-D} provides a novel perspective to study the two-sample covariance testing problem based on the spectrum of the differential matrix $D$, and will be the starting point of constructing our test statistic. 

\paragraph{Notation} 
%We introduce the following notation for the rest of this paper. 
For a vector $v \in \mathbb{R}^p$, let $||v||_q = \left( \sum_{i=1}^p |v_i|^q \right)^{1/q}$ be the $L_q$ norm for $q > 0$, and $||v||_0$ be the number of non-zero elements. For a symmetric matrix $A \in \mathbb{R}^{p \times p}$, let $A_{ij}$ be the $(i, j)$-th element, $||A||_q$ be the $L_q$ norm of vectorized $A$, and ${\rm tr}(A)$ be the trace. In addition, we use $\lambda_1(A) \geq \cdots \geq \lambda_p(A)$ to denote the eigenvalues of $A$. For two symmetric matrices $A, B \in \mathbb{R}^{p \times p}$, we write $A \succeq B$ when $A-B$ is positive semidefinite. Finally, for two sequences of real numbers $\{x_n\}$ and $\{y_n\}$, we write $x_n = O(y_n)$ if $|x_n / y_n| \leq C$ for all $n$ and some positive constant $C$, and $x_n = o(y_n)$ if $\lim_n x_n / y_n = 0$.

%%% Test statistic %%%
\subsection{A two-sample covariance test: sLED}
\label{sec:test-stat}
%Now we propose the sparse-Leading-Eigenvalue-Driven (sLED) test for two-sample covariance matrices.
Starting from \cref{eq:test-D}, note that
\[ \sigma_1(D) = \max \{ \left| \lambda_1( D ) \right|, \, \left| \lambda_1( - D ) \right| \}\,. \]
Therefore, a naive test statistic would be $T^{\textrm{naive}} = \max \left\{ \left| \lambda_1( \hat{D} ) \right|, \, \left| \lambda_1( - \hat{D} ) \right| \right\}$ for some estimator $\hat{D}$. A simple estimator is the difference between the sample covariance matrices:
\begin{equation}
\hat{D} = \hat{\Sigma}_2 - \hat{\Sigma}_1 \,, \\
\textrm{where }
\hat{\Sigma}_1 = \frac{1}{n} \sum_{k=1}^n X_k X_k^T\,, \ 
\hat{\Sigma}_2 = \frac{1}{m} \sum_{l=1}^m Y_l Y_l^T\,.
\label{eq:Dhat}
\end{equation}
However, in the high-dimensional setting, $\lambda_1( \hat{D} )$ is not necessarily a consistent estimator of $\lambda_1(D)$, and without extra assumptions, there is almost no hope of reliable recovery of the eigenvectors \citep{johnstone2009consistency}. A popular remedy for this curse of dimensionality in many high-dimensional methods is to add sparsity assumptions, such as imposing an $L_0$ constraint on an optimization procedure. Note that for any symmetric matrix $A \in \mathbb{R}^{p \times p}$, 
\[
\lambda_1(A) = \max_{||v||_2 = 1} v^T A v = \max_{||v||_2 = 1 } {\rm tr} \left(A (v v^T) \right)\,.
\]
Following the common strategy, we consider the constrained problem:
\begin{equation}
\lambda_1^R(A) =  \max_{||v||_2=1, \, ||v||_0 \leq R} {\rm tr} \left(A (v v^T) \right) \,,
\label{eq:sparse-eigen}
\end{equation}
where $R>0$ is some constant that controls the sparsity of the solution, and $\lambda_1^R(A)$ is usually referred to as the $R$-sparse leading eigenvalue of $A$. Then, naturally, we construct the following test statistic
\begin{equation}
T_R = \max \left\{ \left| \lambda_1^R ( \hat{D} ) \right|, \, \left| \lambda_1^R ( - \hat{D} ) \right| \right\}\,,
\label{eq:test-stat}
\end{equation}
and the sparse-Leading-Eigenvalue-Driven (sLED) test is obtained by thresholding $T_R$ at the proper level.

Problem \cref{eq:sparse-eigen} is closely related with Sparse Principle Component Analysis (SPCA). The only difference is that in SPCA, the input matrix $A$ is usually the sample covariance matrix, but here, we use the differential matrix $\hat{D}$. Solving \cref{eq:sparse-eigen} directly is computationally intractable, but we will show in \Cref{sec:sparsepca} that approximate solutions can be obtained.

Finally, because it is difficult to obtain the limiting distribution of $T_R$, we use a permutation procedure. Specifically, for any $\alpha \in (0, 1)$, the $\alpha-$level sLED test, denoted by $\Psi_\alpha$, is conducted as follows:

\begin{enumerate}
\item Given samples $Z = (X_1, \cdots, X_n, Y_1, \cdots, Y_m)$, calculate the test statistic $T_R$ as in \cref{eq:test-stat}.

\item Sample uniformly from $Z$ without replacement to get $Z^* = (Z_1^*, \cdots, Z_N^*)$, where $N = n+m$. 
\label{step:sample}

\item Calculate the permutation differential matrix $\hat{D}^*$:
\begin{equation}
\hat{D}^* = \hat{\Sigma}_2^* - \hat{\Sigma}_1^*\,,
\label{eq:permutation-cov}
\end{equation}
where 
$\hat{\Sigma}_1^* = \cfrac{1}{n} \sum_{k=1}^n  Z_k^* (Z_k^*)^T \,, \ 
\hat{\Sigma}_2^* = \cfrac{1}{m} \sum_{l=n+1}^N  Z_l^* (Z_l^*)^T$.

\item Compute the permutation test statistic
\label{step:tstar}
\begin{equation}
T_R^* =  \max \left\{ \left| \lambda_1^R ( \hat{D}^* ) \right|, \, \left| \lambda_1^R ( - \hat{D}^* ) \right| \right\}\,.
\label{eq:tstar}
\end{equation}

\item Repeat \crefrange{step:sample}{step:tstar} for $B$ times to get $T_R^{*(1)}, \cdots, T_R^{*(B)}$, then
\[
\hat{p} = \frac{1}{B} \sum_{b=1}^B I_{\{ T_R^{*(b)} > T_R \}}\,,
\]
and sLED rejects $H_0$ if $\hat{p} < \alpha$, i.e.,
$ \Psi_\alpha = I_{\{ \hat{p} < \alpha \}}$.
\end{enumerate}

\paragraph{Remark 1} We can also estimate the support of the $R$-sparse leading eigenvector of $D$, which provides a list of genes that are potentially involved in the disease.  Without loss of generality, suppose $\lambda_1^R(\hat{D}) > \lambda_1^R(-\hat{D})$, we define
\begin{equation}
\textrm{Leverage} \coloneqq \textrm{diag}(\hat{v}\hat{v}^T) = ( \hat{v}_1^2, \cdots,  \hat{v}_p^2)^T\,,
\label{eq:leverage}
\end{equation}
where $\hat{v}$ is the $R$-sparse leading eigenvector of $\hat{D}$ in \cref{eq:sparse-eigen}.
Then the elements with large leverage will be the candidate genes that have altered covariance structure between the two populations. 

%%% Sparse PCA %%%
\subsection{Sparse principle component analysis}
\label{sec:sparsepca}

Many studies on Sparse Principle Component Analysis (SPCA) have provided various algorithms to approximate \cref{eq:sparse-eigen} when $A$ is the sample covariance matrix. Most techniques utilize an $L_1$ constraint to achieve both sparsity and computational efficiency. To name a few, \cite{jolliffe2003modified} form the SCoTLASS problem by directly replacing the $L_0$ constraint by $L_1$ constraint; \cite{zou2006sparse} analyze the problem from a penalized regression perspective; \cite{witten2009penalized} and \cite{shen2008sparse} use the framework of low rank matrix completion and approximation; \cite{d2007direct} and \cite{vu2013fantope} consider the convex relaxation of \cref{eq:sparse-eigen}. 
\textcolor{highlight}{
Recent development of atomic norms also provides an alternative approach to deal with the $L_0$ constrained problems (for example, see \cite{oymak2015simultaneously}).
For the purpose of this paper, we give details of only the following two SPCA algorithms that can be directly generalized to approximate \cref{eq:sparse-eigen} with input matrix $\hat{D}$, the differential matrix.
}

\paragraph{Fantope projection and selection (FPS)}
For a symmetric matrix $A \in \mathbb{R}^{p \times p}$, FPS \citep{vu2013fantope} considers a convex optimization problem:
\begin{equation}
\lambda_{fps}^R ( A ) = \max_{H \in \mathcal{F}^1,\, ||H||_1 \leq R} {\rm tr}(A H)\,,
\label{eq:fps}
\end{equation}
where
$
\mathcal{F}^1 = \{ H \in \mathbb{R}^{p \times p}: {\rm  symmetric, } \ 0 \preceq H \preceq I, \ {\rm tr} (H) = 1 \}
$
is the 1-dimensional Fantope, which is the convex hull of all 1-dimensional projection matrices $\{v v^T: ||v||_2 = 1\}$. In addition, by the Cauchy-Schwarz inequality, if $||v||_2=1$, then $||v v^T ||_1 \leq ||v||_0$. Therefore, \cref{eq:fps} is a convex relaxation of \cref{eq:sparse-eigen}. Moreover, when the input matrix is $\hat{D}$, the problem is still convex, and the ADMM algorithm proposed in \cite{vu2013fantope} can be directly applied. This algorithm has guaranteed convergence, but requires iteratively performing SVD on a $p \times p$ matrix. Moreover, the calculation needs to be repeated $B$ times in the permutation procedure, and becomes computationally demanding when $p$ is on the order of a few thousands. Therefore, we present an alternative heuristic algorithm below, which is much more efficient and typically works well in practice.

\paragraph{Penalized matrix decomposition (PMD)} 
For a general matrix $A \in \mathbb{R}^{p \times p}$, PMD \citep{witten2009penalized} solves a rank-one matrix completion problem:
\begin{equation}
\begin{gathered}
 \lambda_{pmd}^R (A) = \max_{u, v} \textrm{tr} \left( A  \left(u v^T\right) \right)\,, \\
 \textrm{subject to } ||u||_2 \leq 1,\, ||v||_2 \leq 1,\, ||u||_1 \leq \sqrt{R},\, ||v||_1 \leq \sqrt{R}\,.
\end{gathered}
\label{eq:pmd}
\end{equation}
The solution for each one of $u$ and $v$ has a simple closed form after fixing the other one. This leads to a straightforward iterative algorithm, which has been implemented in the R package \texttt{PMA}.
Moreover, if the solutions satisfy $\hat{u}=\hat{v}$, then they are also the solutions to the following non-convex Constrained-PMD problem:
\begin{equation}
\lambda_{c-pmd}^R (A) = \max_{||v||_2 \leq 1,\, ||v||_1 \leq \sqrt{R}} \textrm{tr} \left(A \left( v v^T \right) \right) \,.
\label{eq:aug-pmd}
\end{equation}
Note that the solutions of \cref{eq:aug-pmd} always have $||v||_2=1$, which implies $||v||_1^2 = ||vv^T||_1 \leq ||v||_0$, so \cref{eq:aug-pmd} is also an approximation to \cref{eq:sparse-eigen}.
Now observe that when $A \succeq 0$, as in the usual SPCA setting, the solutions of \cref{eq:pmd} automatically have $\hat{u}=\hat{v}$ by the Cauchy-Schwarz inequality. However, this is no longer true when $A$ is not positive semidefinite, as when we deal with the differential matrix $\hat{D}$. To overcome this issue, we choose some constant $d>0$ that is large enough such that $A + d I \succeq 0$. Then the solutions of $ \lambda_{pmd}^R ( A + d I )$ will satisfy $\hat{u} = \hat{v}$, and it is easy to obtain $\lambda_{c-pmd}^R (A)$ by
\begin{equation}
\lambda_{c-pmd}^R (A) = \lambda_{pmd}^R ( A + d I ) - d\,.
\end{equation}

%%% Consistency %%%
\subsection{Consistency}
\label{sec:consistency}

Finally, we show that sLED is asymptotically consistent. The validity of its size is guaranteed by the permutation procedure. Here, we prove that sLED also achieves full power asymptotically, under the following assumptions:
\begin{enumerate}
\item[(A1)] (Balanced sample sizes) $\underline{c}n \leq m \leq \bar{c} n$ for some constants $0 < \underline{c} \leq 1 \leq \bar{c} < \infty$.

\item[(A2)] (Sub-gaussian tail) Let $(Z_1, \cdots, Z_N) = (X_1, \cdots, X_n, Y_1, \cdots, Y_m)$, then every $Z_k$ is sub-gaussian with parameter $\nu^2$, that is,
\[ \Exp \left[e^{t (Z_k^T u)} \right] \leq e^{\frac{t^2 \nu^2}{2}}, \quad \forall t > 0, \, \forall u \in \mathbb{R}^p \textrm{ such that } ||u||_2=1\,.  \]

\item[(A3)] (Dimensionality) $(\log p)^3 = O(n)$.

\item[(A4)] (Signal strength)  Under $H_1$, for some constant $C$ to be specified later,
\[ \max \left\{ \lambda_1^R (D), \ \lambda_1^R (-D) \right\} \geq C R \sqrt{ \frac{\log p}{n} }\,. \]
\end{enumerate}

\begin{theorem}[Power of sLED] 
\label{thm:power}
Let $T_R$ be the test statistic as defined in \cref{eq:test-stat}, and $T_R^*$ be the permutation test statistic as defined in \cref{eq:tstar}, where $\lambda_1^R(\cdot)$ is approximated by the $L_1$ constrained algorithms \cref{eq:fps} or \cref{eq:aug-pmd}. Then under assumptions (A1)-(A3), for $\forall \delta > 0$, there exists a constant $C$ depending on $(\underline{c}, \bar{c}, \nu^2, \delta)$, such that if assumption (A4) holds, and $n,p$ are sufficiently large,
\[ \Prob_{H_1} \left(  T_R(\hat{D}^*) > T_R(\hat{D}) \right) \leq \delta\,. \]
As a consequence, for any pre-specified level $\alpha \in (0, 1)$, pick $\delta = \alpha / 2$, then
\[  \Prob_{H_1} \left(  \Psi_{\alpha} = 1 \right) \to 1 \ \textrm{ as } B \to +\infty \,. \]
\end{theorem}

The proof of \Cref{thm:power} contains two steps. First, \Cref{thm:permutation} provides an upper bound of the entries in $\hat{D}^*$. Then \Cref{thm:FPS} ensures that the permutation test statistic $T_R^*$ is controlled by $|| \hat{D}^* ||_{\infty}$, and the test statistic $T_R$ is lower-bounded in terms of the signal strength. We state  \Cref{thm:permutation} and \Cref{thm:FPS} below, and the proof details are included in the Supplement.

\begin{theorem}[Permutation differential matrix]
\label{thm:permutation}
Under  assumptions (A1)-(A3), let $\hat{D}^*$ be the permutation differential matrix as defined in \cref{eq:permutation-cov}, then  $\forall \delta > 0$, there exist constants $C$, $C_1$ depending on $(\nu^2, \underline{c}, \bar{c})$, such that if $n, p$ are sufficiently large, 
 \[ \Prob \left( || \hat{D}^* ||_{\infty} > C \sqrt{ \frac{\log (C_1 p^2 / \delta)}{ n }}  \right) \leq \delta\,. \]
\end{theorem}

\begin{theorem}[Test statistic]
\label{thm:FPS}
For any symmetric matrix $\hat{D}$, let $\tilde{\lambda}_1^R(\hat{D})$ be a solution of the $L_1$ constrained algorithms \cref{eq:fps} or \cref{eq:aug-pmd},
then the following statements hold:
\begin{enumerate}[(i)]
\item If $|| \hat{D} ||_\infty \leq \delta$, then $\tilde{\lambda}_1^R(\hat{D}) \leq R \delta$.
\item If there is a matrix $D$ such that $|| \hat{D} - D ||_\infty \leq \delta$, then 
\[ \tilde{\lambda}_1^R (\hat{D}) \geq \lambda_1^R (D) - R \delta\,. \]
\end{enumerate}
\end{theorem}

\paragraph{Remark 2} Assumption (A4) does not require the leading eigenvector of $D$ (or $-D$) to be sparse, only that the sparse signal be strong enough, which is a very mild requirement. In addition, the required sparse signal level, $O \left( R \sqrt{ \log p / n } \right)$, has been shown to be the optimal detection rate for any polynomial-time algorithm under a similar setting \citep{berthet2013optimal}. 

In fact, \cite{berthet2013optimal} also show that without the computational constraint, the optimal signal strength is of order $O \left(\sqrt{ R \log p /n } \right)$. Here, we show that this rate can also be achieved by sLED if we use the exact solutions of the $L_0$ constrained problem \cref{eq:sparse-eigen}. For this purpose, we introduce two slightly different assumptions as follows:

\begin{itemize}
\item[(A3')] (Dimensionality) $R = o(p), \ R^5 (\log p)^3 = O(n)$.

\item[(A4')] (Signal strength)  Under $H_1$, for some constant $C$ to be specified later,
\[ \max \left\{ \lambda_1^R (D), \ \lambda_1^R (-D) \right\} \geq C \sqrt{ \frac{R \log p}{n} }\,. \]
\end{itemize}

Now we state the results regarding the $L_0$ constrained solutions, and the proofs are included in the Supplement.

\begin{theorem}[Power of sLED without computational constraint] 
\label{thm:power-l0}
Let $T_R$ be the test statistic as defined in \cref{eq:test-stat} with $\lambda_1^R(\cdot)$ being a global optimum of \cref{eq:sparse-eigen}. Then under assumptions (A1)-(A2) and (A3'), for any pre-specified level $\alpha \in (0, 1)$, there exists a constant $C$ depending on $(\underline{c}, \bar{c}, \nu^2, \alpha)$, such that if assumption (A4') holds, and $n, p$ are sufficiently large, 
\[  \Prob_{H_1} \left(  \Psi_{\alpha} = 1 \right) \to 1 \  \textrm{ as } B \to +\infty \,. \]
\end{theorem}

\paragraph{Remark 3} Recall the toy example in \cref{eq:toy-ex}. If we let $R=s$, then the $R$-sparse leading eigenvalue of $D$ is $\lambda_1^R (D) = s \delta$, and sLED remains powerful for $\delta$ as small as $O \left(\sqrt{ \log p / (ns) } \right)$. On the other hand, the maximal entry method \citep{cai2013two} cannot succeed under this setting since it requires $\delta$ to be of order $O \left(\sqrt{ \log p / n } \right)$ or higher.

\textcolor{highlight}{
\paragraph{Remark 4} One might notice that under the toy example in \cref{eq:toy-ex}, assumption (A4) in \Cref{thm:power} for the $L_1$ constrained sLED implies $\delta \geq C \sqrt{\log p / n}$, which is also the required rate for the maximal entry method \citep{cai2013two}. However, we shall view the theoretical results in \Cref{thm:power} as a ``sanity check", in the sense that sLED is at least as good as the maximal entry test.  We know that the maximal entry test will fail when $\delta$ is much smaller than $\sqrt{\log p / n}$, while our theory says that sLED will succeed whenever the maximal entry method succeeds.  This is a worst case guarantee for sLED. In practice, sLED will often output $L_0$-sparse solutions, and \Cref{thm:power-l0} demonstrates the potential of sLED when the solution also happens to be $L_0$-sparse.
}

\subsection{Choosing sparsity parameter $R$}
\label{sec:choose-r}
The tuning parameter $R$ in  \cref{eq:fps} and \cref{eq:aug-pmd} plays an important role in sLED test.  If $R$ is too large, the method uses little regularization and assumption (A4) is unlikely to hold.  If $R$ is too small, then the constraint is too strong to cover the signal in the differential matrix.  The practical success of sLED requires an appropriate choice of $R$.  We know that $R$ provides a natural, but possibly loose, lower bound on the support size of the estimated sparse eigenvector. In general, one can use cross-validation to choose $R$, so that the estimated leading  sparse singular vector maximizes its inner product with a differential matrix computed from a testing subsample.

In applications, one can often choose $R$ with the aid of subject background knowledge and the context of subsequent analysis.  For example, in the detection of Schizophrenia risk genes, we typically expect to report a certain proportion in a collection of genes for further investigation.  Thus one can choose from a set of candidate values of $R$ to match the desired number of discoveries. 
\textcolor{highlight}{
In this paper, following \cite{witten2009penalized}, we use algorithm \cref{eq:aug-pmd} and choose the sparsity parameter $R$ to be 
\begin{equation}
\sqrt{R} = c \sqrt{p}\,, ~ \textrm{for some}~ c \in (0, 1)\,,
\label{eq:r-c}
\end{equation}
then $c^2$ provides a loose lower bound on the proportion of selected genes. We will illustrate in simulation studies (\Cref{sec:simulations}) and the CMC data application (\Cref{sec:data}) that sLED is stable with a reasonable range of $c$. 
}

%!TEX root = ./CovTest-sPCA.tex

\section{Simulations}
\label{sec:simulations}

In this section, we conduct simulation studies to compare the power of \texttt{sLED} with other existing methods:
\textcolor{highlight}{
 \cite{schott2007test} use an estimator of the Frobenius norm $||D||_F^2$ (\texttt{Sfrob}); \cite{li2012two} use a linear combination of three U-statistics which is also motivated by $||D||_F^2$ (\texttt{Ustat});
 }
%\cite{srivastava2010testing} that is based on the quantity $\textrm{tr}(\Sigma_1^2) / (\textrm{tr}(\Sigma_1))^2 - \textrm{tr}(\Sigma_2^2) / (\textrm{tr}(\Sigma_2))^2$ (\texttt{TRdiff}), 
\cite{cai2013two} use the maximal absolute entry of $D$ (\texttt{Max}); \cite{chang2015bootstrap} use a multiplier bootstrap on the same test statistic (\texttt{MBoot}), and \cite{wu2015tests} use random matrix projections (\texttt{RProj}).
%as well as \cite{szekely2013energy}  that use energy statistics to measure the distance between $\Sigma_1, \Sigma_2$, and use permutation to obtain $p$-values (\texttt{Energy}).
\textcolor{highlight}{
To obtain a fair comparison of empirical power, we use permutation to compute the $p$-values for all methods, except for \texttt{MBoot}, which already uses a bootstrap procedure. Because the empirical size is properly controlled by permutation, we focus on comparing the empirical power in the rest of this section.
%For completeness, since the correlation matrices under $H_0$ and $H_1$ are all different in our simulation scenarios, we also report the sizes and power of the two-sample correlation matrix test proposed by \cite{cai2015inference} that uses a similar test statistic as \texttt{Max} (\texttt{MaxCorr}) for comparison. 
The simulation results in this section can be reproduced using the code provided at \url{https://github.com/lingxuez/sLED}.
}

We consider four different covariance structures of $\Sigma_1$ and $\Sigma_2 = \Sigma_1 + D$ under the alternative hypothesis. 
Under each scenario $i=1, \cdots, 4$, we first generate a base matrix $\Sigma^{*(i)}$, and we enforce positive definiteness using $\Sigma_1 = \Sigma^{*(i)} + \delta I_p$ and $\Sigma_2 =  \Sigma^{*(i)} + D + \delta I_p$, where $\delta = \left|\min \{\lambda_{\min}(\Sigma^{*(i)}),  \lambda_{\min}(\Sigma^{*(i)}+D)\} \right| + 0.05$.
Now we specify the structures of $\{\Sigma^{*(i)}\}_{i=1, ..., 4}$. 
 Under each scenario, we let $\Lambda \in \mathbb{R}^{p \times p}$ to be a diagonal matrix with diagonal elements being sampled from $\textrm{Unif}(0.5, 2.5)$ independently. We denote $\lfloor x \rfloor$ to be the largest integer that is smaller than or equal to $x$.

\begin{enumerate}
\item {\bf Noisy diagonal}. \textcolor{highlight}{
Let $\Delta^{(1)}_{ii}=1$, $\Delta^{(1)}_{ij} \sim \textrm{Bernoulli}(0.05)$ when $i<j$, and $\Delta^{(1)}_{ij} = \Delta^{(1)}_{ji}$ when $i > j$ for symmetry, and we define $\Sigma^{*(1)} = \Lambda^{1/2} \Delta^{(1)}  \Lambda^{1/2}$. This model is also considered in \cite{cai2013two}.
}

\item {\bf Block diagonal}. Let $K=\lfloor p/10 \rfloor$ be the number of blocks, $\Delta^{(2)}_{ii}=1$, $\Delta^{(2)}_{ij}=0.55$ when $10(k-1)+1 \leq i \neq j \leq 10k$ for $k=1, \cdots, K$, and zero otherwise. We define $\Sigma^{*(2)} = \Lambda^{1/2} \Delta^{(2)}  \Lambda^{1/2}$. This model is also considered in \cite{cai2013two} and \cite{chang2015bootstrap}.

\item {\bf Exponential decay}. Let $\Delta^{(3)}_{ij} = 0.5^{|i-j|}$, and $\Sigma^{*(3)} = \Lambda^{1/2} \Delta^{(3)}  \Lambda^{1/2}$. This model is also considered in \cite{cai2013two} and \cite{chang2015bootstrap}.

\item {\bf WGCNA}. \textcolor{highlight}{
Here we construct $\Sigma^{*(4)}$ based on the CMC data \citep{CMC2016} using the simulation tool provided by WGCNA \citep{zhang2005general}. Specifically, we first compute the eigengene (i.e., the first principal component) of the {\it M2c} module for the 279 control samples. The {\it M2c} module will be the focus of \Cref{sec:data}, and more detailed discussion is provided there. We use the {\tt simulateDatExpr} command in the {\tt WGCNA} R package to simulate new expressions for $p$ genes of the 279 samples. We set {\tt modProportions=(.8,.2)}, such that 80\% of the $p$ genes are simulated to be correlated with the {\it M2c} eigengene, and the other $20\%$ genes are randomly generated. Default values are used for all other parameters. Finally, $\Sigma^{*(4)}$ is set to be the sample covariance matrix.
}
\end{enumerate}

We consider the following two types of differential matrix $D$:
\begin{enumerate}
\item {\bf Sparse block difference}. Suppose $D$ is supported on an $s \times s$ sub-block with $s = \lfloor 0.1 p \rfloor$, 
\textcolor{highlight}{
and the non-zero entries are generated from $\textrm{Unif}(d / 2, 2d)$ independently.  
The signal level $d$ is chosen to be 
$
d = \frac{1}{2} \sqrt{ \left(  \max_{1\leq j \leq p} \Sigma_{jj}^{*} \right) \log(p)}
$
 where $\Sigma^{*}$ is the base matrix defined above.}

\item {\bf Soft-sparse spiked difference}. 
Let $D$ be a rank-one matrix with $D = d v v^T$, 
where $v$ is a soft-sparse unit vector with $||v||_2 = 1$ and $||v||_0 = \lfloor 0.2p \rfloor$.  The support of $v$ is uniformly sampled from $\{1, \cdots, p\}$ without replacement. 
\textcolor{highlight}{
Among the non-zero elements, $\lfloor 0.1p \rfloor$ are sampled from $N(1, 0.01)$, and the remaining  $\lfloor 0.2p \rfloor - \lfloor 0.1p \rfloor$ are sampled from $N(0.1, 0.01)$. Finally, $v$ is normalized to have unit $L_2$ norm. 
The signal level $d$ is set to be
$
d = 4  \sqrt{ \left(  \max_{1\leq j \leq p} \Sigma_{jj}^{*} \right) \log(p)}
$, where $\Sigma^{*}$ is the base matrix defined above. The differential matrix $D$ under this scenario is moderately sparse, with $\lfloor 0.1p \rfloor$ features exerting larger signals.
}
\end{enumerate}

Finally, the samples are generated by $X_i = \Sigma_1^{1/2} Z_i$ for $i = 1, \cdots, n$, and $Y_l = \Sigma_2^{1/2} Z_{n+l}$ for $l = 1, \cdots, m$, where $\{Z_i\}_{i=1, n+m}$ are independent $p$-dimensional random variables with $i.i.d.$ coordinates $Z_{ij}$, $j=1, \cdots, p$. We consider the following four distributions for $Z_{ij}$:
\begin{enumerate}
\item Standard Normal $N(0, 1)$.
\item \textcolor{highlight}{
Centralized Gamma distribution with $\alpha=4, \beta=0.5$ (i.e., the theoretical expectation $\alpha \beta=2$ is subtracted from $\Gamma(4, 0.5)$ samples). This distribution is also considered in \cite{li2012two} and \cite{cai2013two}.
}
\item \textcolor{highlight}{
$t$-distribution with degrees of freedom 12. This distribution is also considered in \cite{cai2013two} and \cite{chang2015bootstrap}.
}
\item \textcolor{highlight}{
Centralized Negative Binomial distribution with mean $\mu=2$ and dispersion parameter $\phi=2$ (i.e., the theoretical expectation $\mu=2$ is subtracted from NB$(2, 2)$ samples). 
}
\end{enumerate}
Note that when $Z_{ij} \sim N(0, 1)$, $X$ and $Y$ are multinomial Gaussian random variables with covariance matrices $\Sigma_1$ and $\Sigma_2$. We also consider three non-Gaussian distributions to account for the heavy-tail scenario observed in many genetic data sets.

Here, the smoothing parameter for sLED is set to be $\sqrt{R} = 0.3 \sqrt{p}$, and 100 random projections are used for {\tt Rproj}. For {\tt sLED}, {\tt Max}, {\tt Ustat}, {\tt Sfrob}, and {\tt Rproj}, 100 permutations are used to obtain each $p$-value; for {\tt MBoot}, 100 bootstrap repetitions are used. 
\textcolor{highlight}{
\Cref{tab:sim-power} summarizes the empirical power under different covariance structures and differential matrices when $Z_{ij}$'s are sampled from standard Normal and centralized Gamma distribution. We see that {\tt sLED} is more powerful than many existing methods under most scenarios. The results using the other two distributions of $Z$ have similar patterns, and due to space limitation we include them in the Supplement. 
}

\textcolor{highlight}{
Finally, we examine the sensitivity of sLED to the smoothing parameter. Recall that the smoothing parameter is set to be $\sqrt{R} = c \sqrt{p}$ as explained in \Cref{sec:choose-r}. \cref{fig:sim-sled} visualizes the empirical power of sLED using $c \in \{0.10, 0.12, \cdots, 0.30\}$ when $D$ has sparse block difference and $Z_{ij}$'s are sampled from $N(0, 1)$. It is clear that sLED remains powerful for a wide range of $c$. Similar patterns are observed under other scenarios, and we include these results in the Supplement.
}

\bgroup
\def\arraystretch{1.4} % table vertical spacing; 1 is the default
\begin{table}[htbp]
\begin{center}
\caption{Empirical power in 100 repetitions, where $n=m=100$, nominal level $\alpha = 0.05$, and $Z_{ij}$'s are sampled from standard normal (top) and centralized Gamma $(4, 0.5)$ (bottom). Under each scenario, the largest power is highlighted.
}
\label{tab:sim-power}

\begin{adjustbox}{center}
\begin{tabular}{l l |  l l l |  l l l |  l l l |  l l l }
\hline
$\mathbf{D}$ &
 \multicolumn{1}{l}{$\mathbf{\Sigma_1}$}
& \multicolumn{3}{c}{\bf Noisy diagonal} & \multicolumn{3}{c}{\bf Block diagonal} & \multicolumn{3}{c}{\bf Exp. decay} & \multicolumn{3}{c}{\bf WGCNA} \\
\cline{3-5}\cline{6-8}\cline{9-11}\cline{12-14}
 & \multicolumn{1}{l}{$\mathbf{p}$} 
 & {\bftab 100} & {\bftab 200} & \multicolumn{1}{c}{\bftab 500}  
 & {\bftab 100} & {\bftab 200} & \multicolumn{1}{c}{\bftab 500}  
 & {\bftab 100} & {\bftab 200} & \multicolumn{1}{c}{\bftab 500}   
 & {\bftab 100} & {\bftab 200} & \multicolumn{1}{c}{\bftab 500}   \\
 \hline
 & \multicolumn{1}{c}{} &  \multicolumn{12}{c}{Gaussian} \\
{\bf Block}
& Max & 0.38 & 0.14 & 0.11 & 0.94 & 0.54 & 0.25 & 0.98 & 0.86 & 0.31 & 0.92 & 0.64 & 0.16 \\ 
& MBoot & 0.39 & 0.18 & 0.13 & 0.94 & 0.54 & 0.31 & 0.98 & 0.88 & 0.30 & 0.89 & 0.63 & 0.20 \\ 
&  Ustat & 0.71 & 0.66 & 0.74 & 0.98 & 0.96 & 0.95 & {\bftab 1.00} & {\bftab 1.00} & 0.99 & 0.76 & 0.78 & 0.85 \\ 
& Sfrob & 0.72 & 0.64 & 0.73 & 0.97 & 0.95 & 0.95 & {\bftab 1.00} & {\bftab 1.00} & 0.99 & 0.72 & 0.79 & 0.86 \\ 
 %& Energy & 0.05 & 0.06 & 0.06 & 0.10 & 0.06 & 0.06 & 0.05 & 0.05 & 0.08 & 0.07 & 0.06 & 0.04 \\ 
&  RProj & 0.09 & 0.13 & 0.09 & 0.13 & 0.16 & 0.14 & 0.24 & 0.16 & 0.09 & 0.20 & 0.17 & 0.06 \\ 
&  sLED &{\bftab 1.00} & {\bftab 0.99} & {\bftab 1.00} & {\bftab 1.00} & {\bftab 0.99} & {\bftab 1.00} & {\bftab 1.00} & {\bftab 1.00} & {\bftab 1.00} & {\bftab 0.98} & {\bftab 0.96} & {\bftab 0.95} \\ 
  \hline
{\bf Spiked}
& Max & 0.12 & 0.08 & 0.05 & 0.49 & 0.26 & 0.09 & 0.96 & 0.90 & 0.15 & 0.86 & 0.32 & 0.04 \\ 
&  MBoot & 0.12 & 0.08 & 0.05 & 0.51 & 0.29 & 0.11 & 0.98 & 0.90 & 0.17 & 0.79 & 0.31 & 0.07 \\ 
&  Ustat & 0.20 & 0.11 & {\bftab 0.13} & 0.76 & 0.44 & 0.06 & {\bftab 1.00} & 0.95 & 0.60 & 0.30 & 0.10 & 0.04 \\ 
&  Sfrob & 0.18 & {\bftab 0.12} & 0.11 & 0.73 & 0.41 & 0.07 & {\bftab 1.00} & 0.93 & 0.62 & 0.34 & 0.14 & 0.03 \\ 
%&  Energy & 0.06 & 0.03 & 0.06 & 0.07 & 0.11 & 0.06 & 0.05 & 0.10 & 0.09 & 0.10 & 0.03 & 0.04 \\ 
 & RProj  & 0.10 & 0.08 & 0.02 & 0.32 & 0.08 & \bftab 0.12 & 0.30 & 0.20 & 0.09 & 0.61 & 0.24 & \bftab 0.13 \\ 
 & sLED & {\bftab 0.51} & 0.11 & 0.03 & {\bftab 0.97} & {\bftab 0.70} & {\bftab 0.12} & {\bftab 1.00} & {\bftab 1.00} & {\bftab 1.00} & {\bftab 0.97} & {\bftab 0.57} & 0.05 \\ 
  \hline
     & \multicolumn{1}{c}{} &  \multicolumn{12}{c}{Centralized Gamma} \\
{\bf Block}
&  Max & 0.42 & 0.20 & 0.14 & 0.89 & 0.71 & 0.28 & 0.96 & 0.82 & 0.42 & 0.77 & 0.67 & 0.27 \\ 
&  MBoot & 0.42 & 0.14 & 0.10 & 0.86 & 0.58 & 0.20 & 0.95 & 0.77 & 0.33 & 0.72 & 0.63 & 0.25 \\ 
&  Ustat & 0.57 & 0.59 & 0.70 & 0.92 & 0.94 & 0.97 & 0.99 & 0.98 & 0.98 & 0.53 & 0.82 & 0.86 \\ 
&  Sfrob & 0.58 & 0.55 & 0.71 & 0.92 & 0.92 & 0.98 & 0.99 & 0.99 & 0.98 & 0.50 & 0.76 & 0.81 \\ 
 % & Energy & 0.03 & 0.04 & 0.11 & 0.06 & 0.03 & 0.10 & 0.11 & 0.09 & 0.06 & 0.04 & 0.02 & 0.06 \\ 
&  RProj & 0.11 & 0.10 & 0.09 & 0.24 & 0.17 & 0.16 & 0.41 & 0.15 & 0.14 & 0.21 & 0.07 & 0.14 \\ 
&  sLED & \bftab 0.96 & \bftab 0.98 & \bftab 0.99 & \bftab 0.99 & \bftab 1.00 & \bftab 1.00 & \bftab 1.00 & \bftab \bftab 1.00 & \bftab 1.00 &\bftab  0.94 & \bftab 0.88 & \bftab 0.94 \\ 
\hline
{\bf Spiked}
&  Max & 0.08 & 0.09 & 0.03 & 0.72 & 0.39 & 0.05 & 0.99 & 0.71 & 0.22 & 0.91 & 0.35 & 0.04 \\ 
&  MBoot & 0.10 & 0.07 & 0.02 & 0.74 & 0.36 & 0.09 & 0.99 & 0.71 & 0.16 & 0.88 & 0.35 & 0.04 \\ 
&  Ustat & 0.32 & 0.08 & \bftab 0.11 & 0.78 & 0.41 & 0.07 & \bftab 1.00 & 0.94 & 0.70 & 0.33 & 0.08 & 0.05 \\ 
&  Sfrob & 0.34 & 0.07 & 0.10 & 0.80 & 0.37 & 0.07 & \bftab 1.00 & 0.96 & 0.74 & 0.28 & 0.04 & 0.04 \\ 
%&  Energy & 0.07 & 0.06 & 0.06 & 0.04 & 0.04 & 0.03 & 0.11 & 0.12 & 0.03 & 0.05 & 0.10 & 0.08 \\ 
&  RProj & 0.12 & 0.06 & 0.07 & 0.32 & 0.12 & 0.06 & 0.36 & 0.14 & 0.10 & 0.63 & 0.22 & \bftab 0.10 \\  
&  sLED &\bftab  0.56 &\bftab  0.14 & 0.05 & \bftab 0.97 & \bftab 0.71 & \bftab 0.14 & \bftab 1.00 & \bftab 1.00 & \bftab 1.00 & \bftab 0.93 &\bftab  0.51 & 0.08 \\ 
\hline
\end{tabular}
\end{adjustbox}

\end{center}
\end{table}%

\begin{figure}[htbp]	
	\centering
	\begin{subfigure}[b]{0.24\textwidth}
		\centering
		\includegraphics[width=\textwidth]{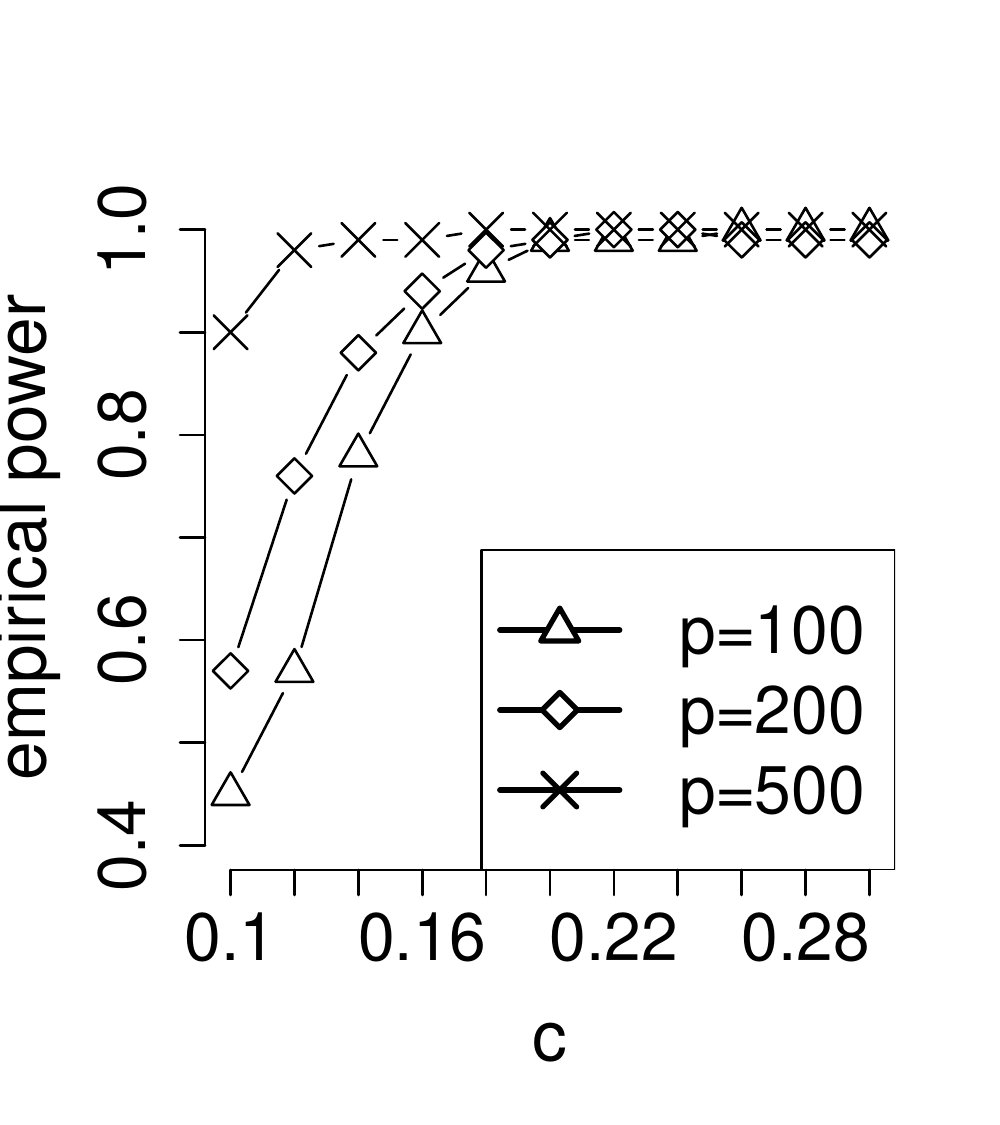}
		\caption{Noisy diagonal}
	\end{subfigure}
	\begin{subfigure}[b]{0.24\textwidth}
		\centering
		\includegraphics[width=\textwidth]{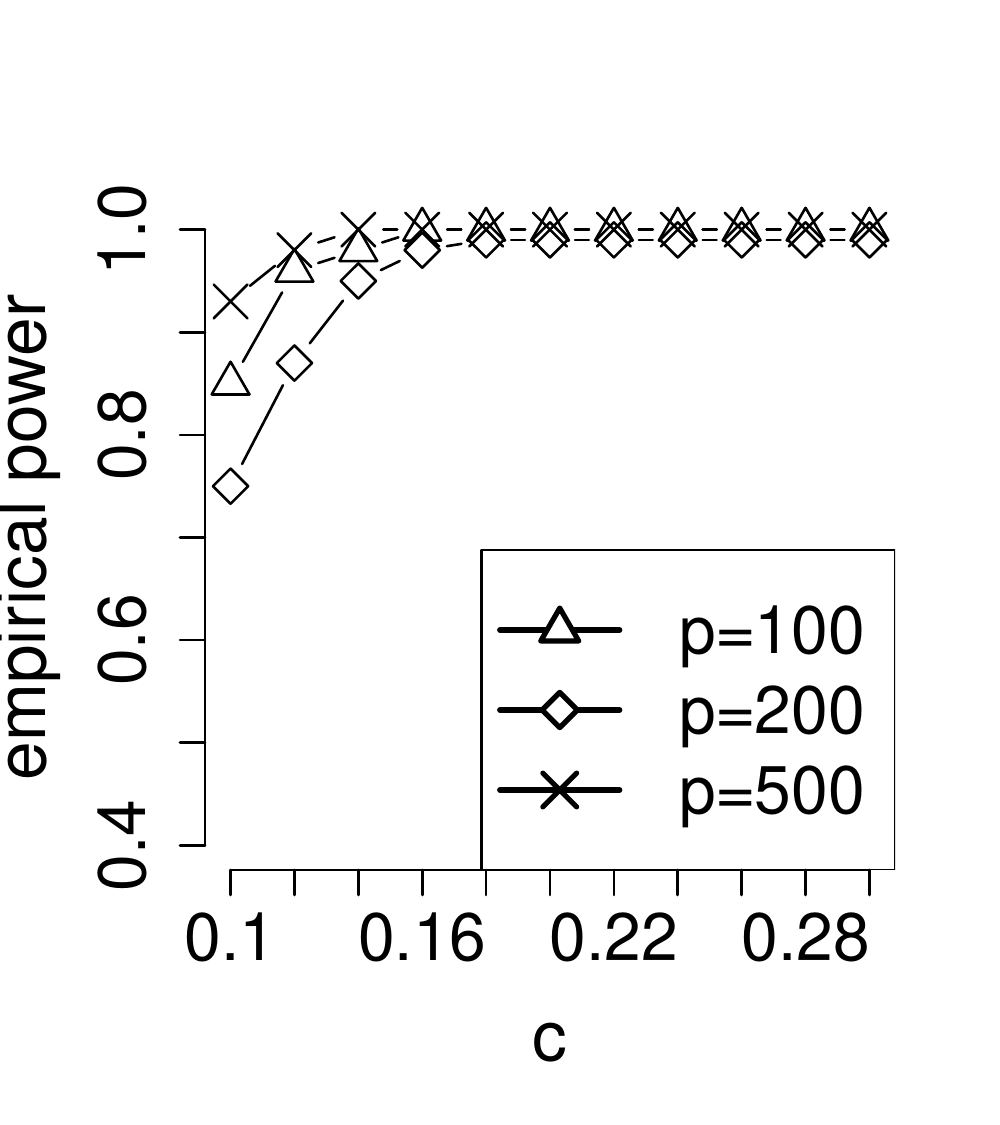}
		\caption{Block diagonal}
	\end{subfigure}
	\begin{subfigure}[b]{0.24\textwidth}
		\centering
		\includegraphics[width=\textwidth]{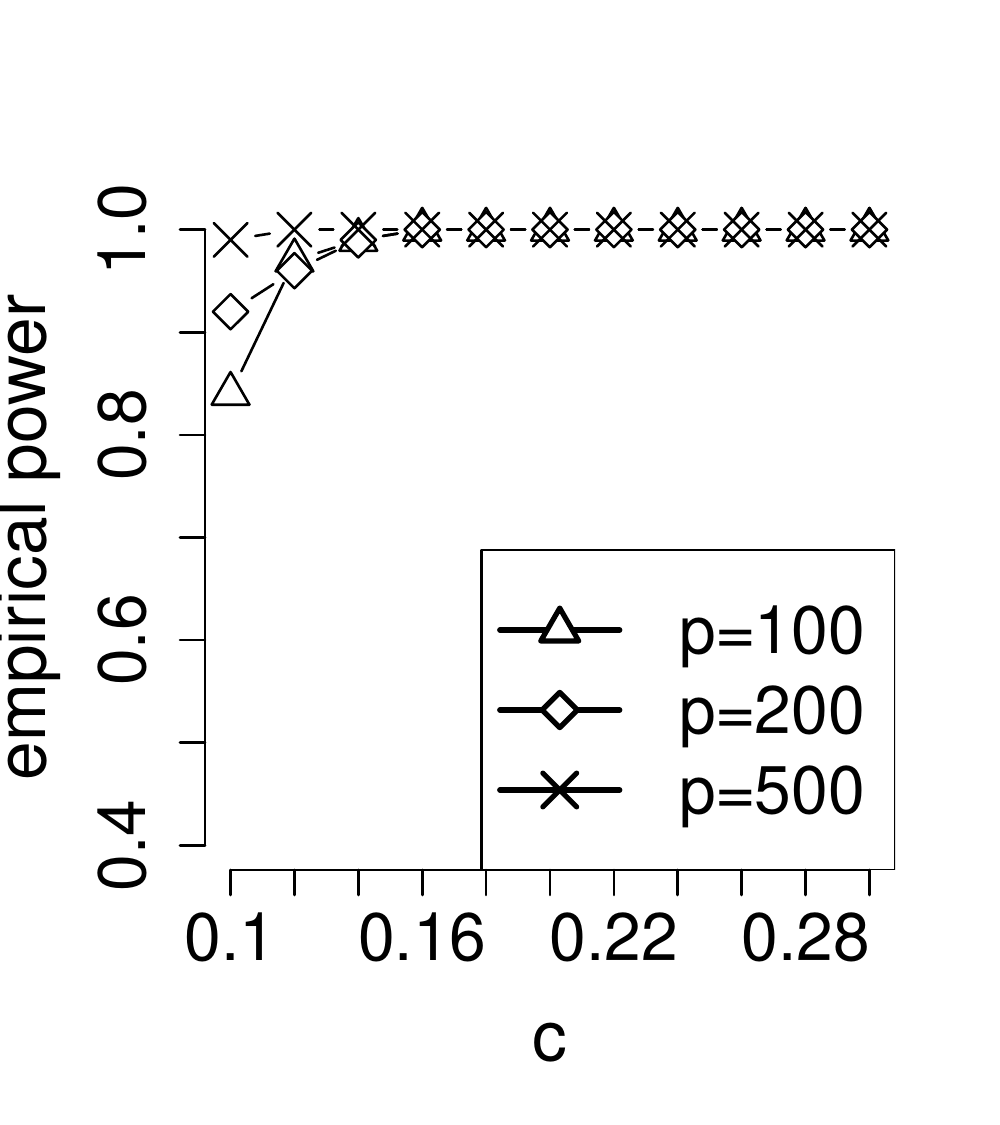}
		\caption{Exp. decay}
	\end{subfigure}
	\begin{subfigure}[b]{0.24\textwidth}
		\centering
		\includegraphics[width=\textwidth]{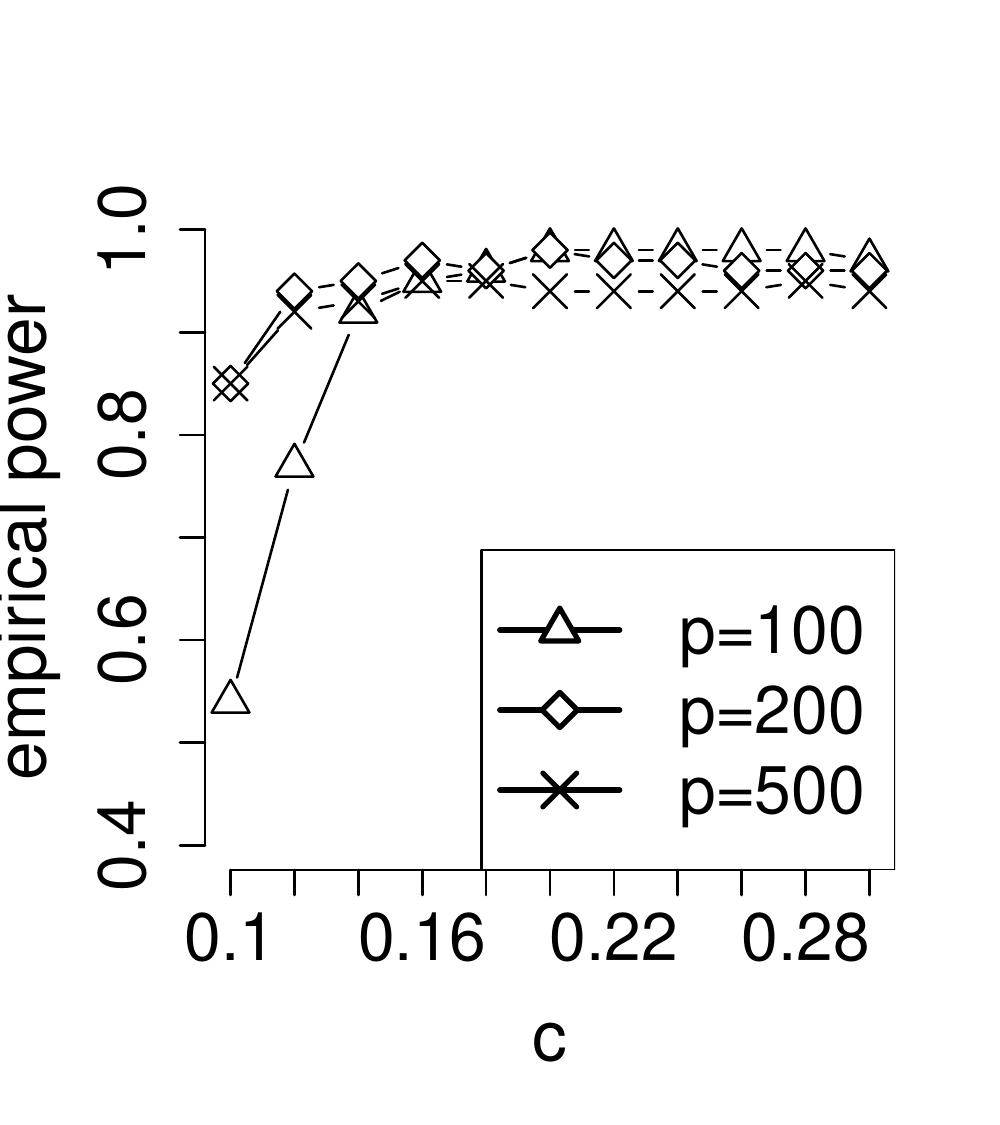}
		\caption{WGCNA}
	\end{subfigure}

	\caption{Empirical power of sLED in 100 repetitions using different smoothing parameters $\sqrt{R} = c \sqrt{p}$ for $c \in \{0.10, \cdots, 0.30\}$, where $D$ has sparse block difference and $Z_{ij}$'s are sampled from $N(0, 1)$.}
	\label{fig:sim-sled}
\end{figure}

%!TEX root = CovTest-sPCA.tex

\section{Application to Schizophrenia data}
\label{sec:data}

In this section, we apply sLED to the CommonMind Consortium (CMC) data, containing RNA-sequencing on 16,423 genes from 258 Schizophrenia (SCZ) subjects and 279 control samples \citep{CMC2016}. 
\textcolor{highlight}{
The RNA-seq data has been carefully processed, including log-transformation and correction of various covariates. The CMC group further cluster the genes into 35 genetic modules using the WGCNA tool \citep{zhang2005general}, such that genes within each module tend to be closely connected and have related biological functionalities. Among these, the {\it M2c} module, containing 1,411 genes, is the only one that is enriched with genes exhibiting differential expression and with prior genetic associations with schizophrenia (SCZ). We direct readers to the original paper \citep{CMC2016} for more detailed description of the data processing and genetic module analysis. In the rest of this section, we apply sLED to investigate the co-expression differences between cases and controls in the {\it M2c} module, which is of the greatest scientific interest.  
We center and standardize the expression data, such that each gene has mean 0 and standard deviation 1 across samples. Therefore, the covariance test is applied to correlation matrices.
}

\subsection{Testing co-expression differences}
\label{sec:data-cov}
In this section, we use sLED to compare the correlation matrices among the 1,411 {\it M2c}-module genes between SCZ and control samples. The sparsity parameter in \cref{eq:aug-pmd} is chosen to be $\sqrt{R} = c \sqrt{p}$ as explained in \Cref{sec:choose-r}. Here, because the number of risk genes that carry the genetic signals is expected to be roughly in the range of $1\%$--$10\%$, we choose $c=0.1$. Applying sLED with 1,000 permutation repetitions, we obtain a $p$-value of 0.014, indicating a significant difference between SCZ and control samples.

We then identify the key genes that drive this difference according to their leverage, as defined in \cref{eq:leverage}. Specifically, we order the leverage of all genes, such that
$ \hat{v}_{(1)}^2 \geq \hat{v}_{(2)}^2 \geq \cdots \geq \hat{v}_{(p)}^2$, where larger leverage usually indicates stronger signals. Note that by construction, $\sum_{i=1}^p \hat{v}_{(i)}^2 = 1$. Among the 1,411 genes, 113 genes have non-zero leverage, and we call them {\it top genes}. Moreover, we notice that the first 25 genes have already achieved a cumulative leverage of $0.999$, so we refer to them as {\it primary genes}. The remaining 88 top genes account for the rest $0.001$ leverage and are referred to as {\it secondary genes} (see \cref{fig:dhat-scree} for the visualization of this cut-off in a scree plot). We show in \cref{fig:dhat-lev} how these 113 top genes form a clear block structure in the differential matrix $\hat{D} = \hat{\Sigma}_{control} - \hat{\Sigma}_{SCZ}$. Notably, such a block structure cannot be revealed if ordered by the differentially expressed $p$-values (\cref{fig:dhat-DE}).

\begin{figure}[t!]
\centering

\begin{subfigure}[t]{0.31\textwidth}
        \centering
        \includegraphics[width = \textwidth]{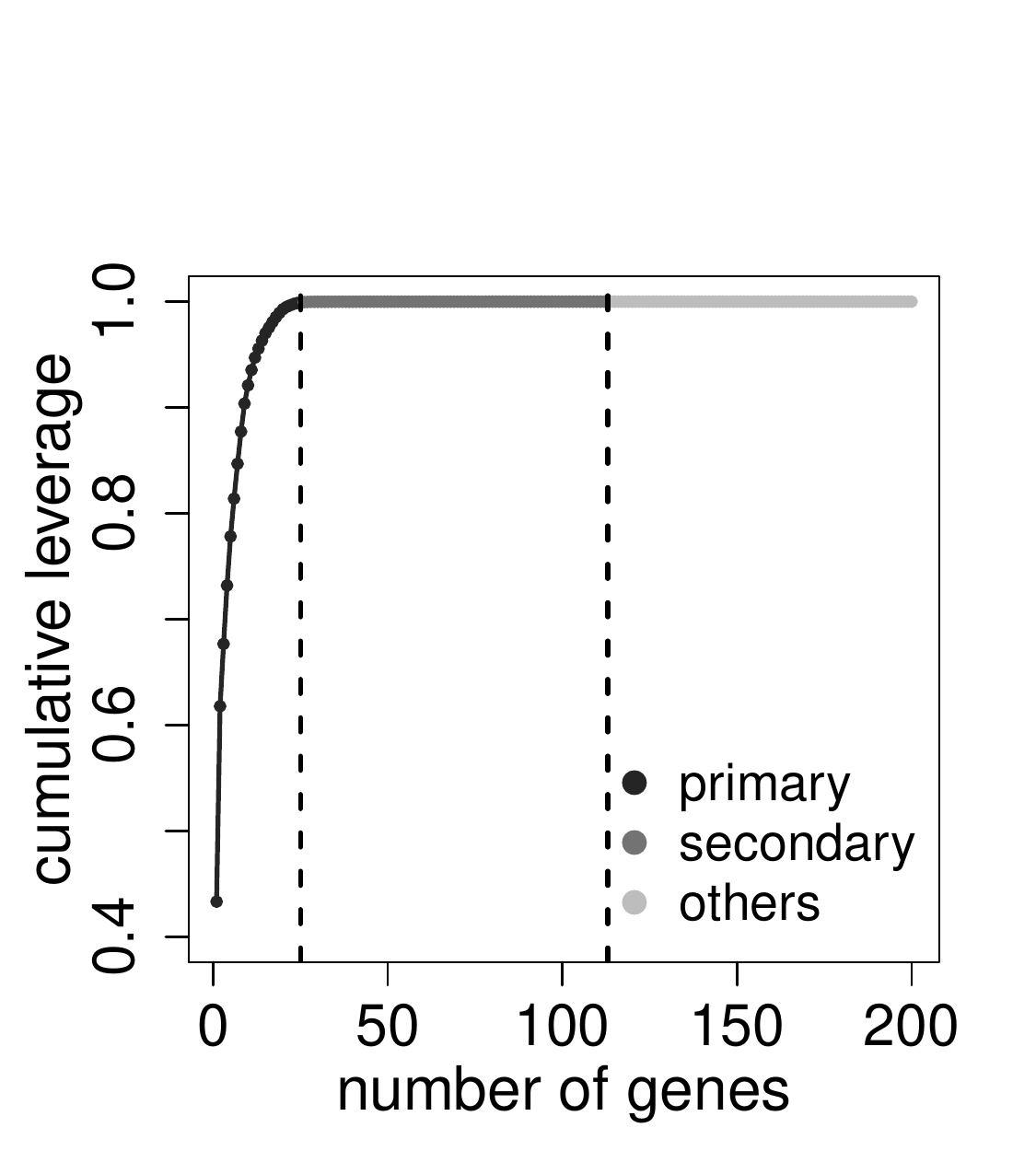}
	\caption{}\label{fig:dhat-scree}
\end{subfigure}
~
\begin{subfigure}[t]{0.32\textwidth}
        \centering
        \includegraphics[width = \textwidth]{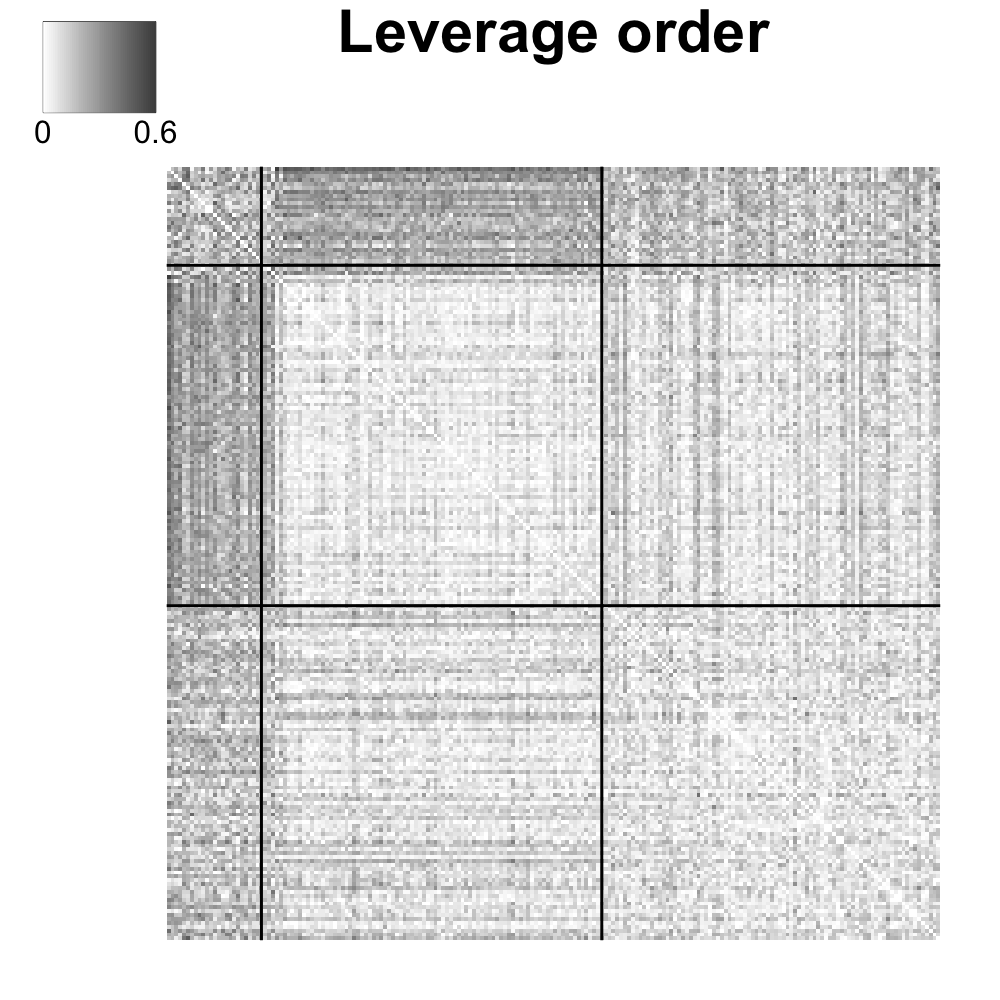}
	\caption{}\label{fig:dhat-lev}
\end{subfigure}
~
\begin{subfigure}[t]{0.32\textwidth}
        \centering
        \includegraphics[width = \textwidth]{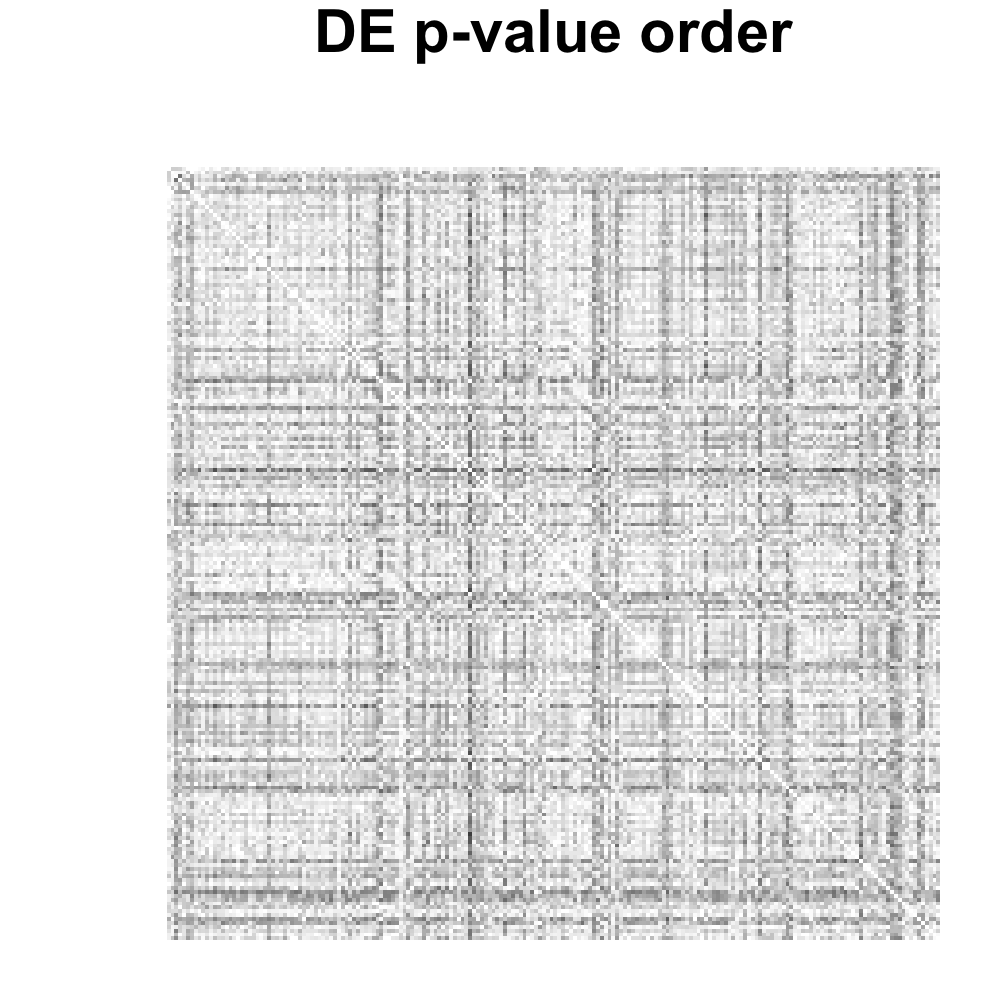}
	\caption{}\label{fig:dhat-DE}
\end{subfigure}

\caption{Visualization of 200 genes in the {\em M2c} module, including 25 primary genes that account for a total leverage of 0.999, 88 secondary genes that account for the remaining 0.001 leverage, and 87 randomly selected other genes that have zero leverage.
{\em (a)} Scree plot of cumulative leverage.
{\em (b)} Heatmap of $|\hat{D}|$ where genes are ordered by leverage and a block structure is revealed. The two partitioning lines indicate the 25 primary genes and the 88 secondary genes.
{\em (c)} Heatmap of $| \hat{D} |$ where genes are ordered by $p$-values in differential expression analysis. Now the block structure is diluted.
}

\label{fig:dhat}
\end{figure}

\cref{fig:dhat-lev} reveals a significant decrease of gene co-expression (interactions) in cortical samples from SCZ subjects between the 25 primary genes and the 88 secondary genes. This pattern is more clearly illustrated in \cref{fig:blue-network}, where two gene networks are constructed for these 113 top genes in control samples and SCZ samples separately (see \Cref{tab:main-rest} for gene names).

\begin{figure}[htbp]
\begin{center}
	\includegraphics[width = \textwidth]{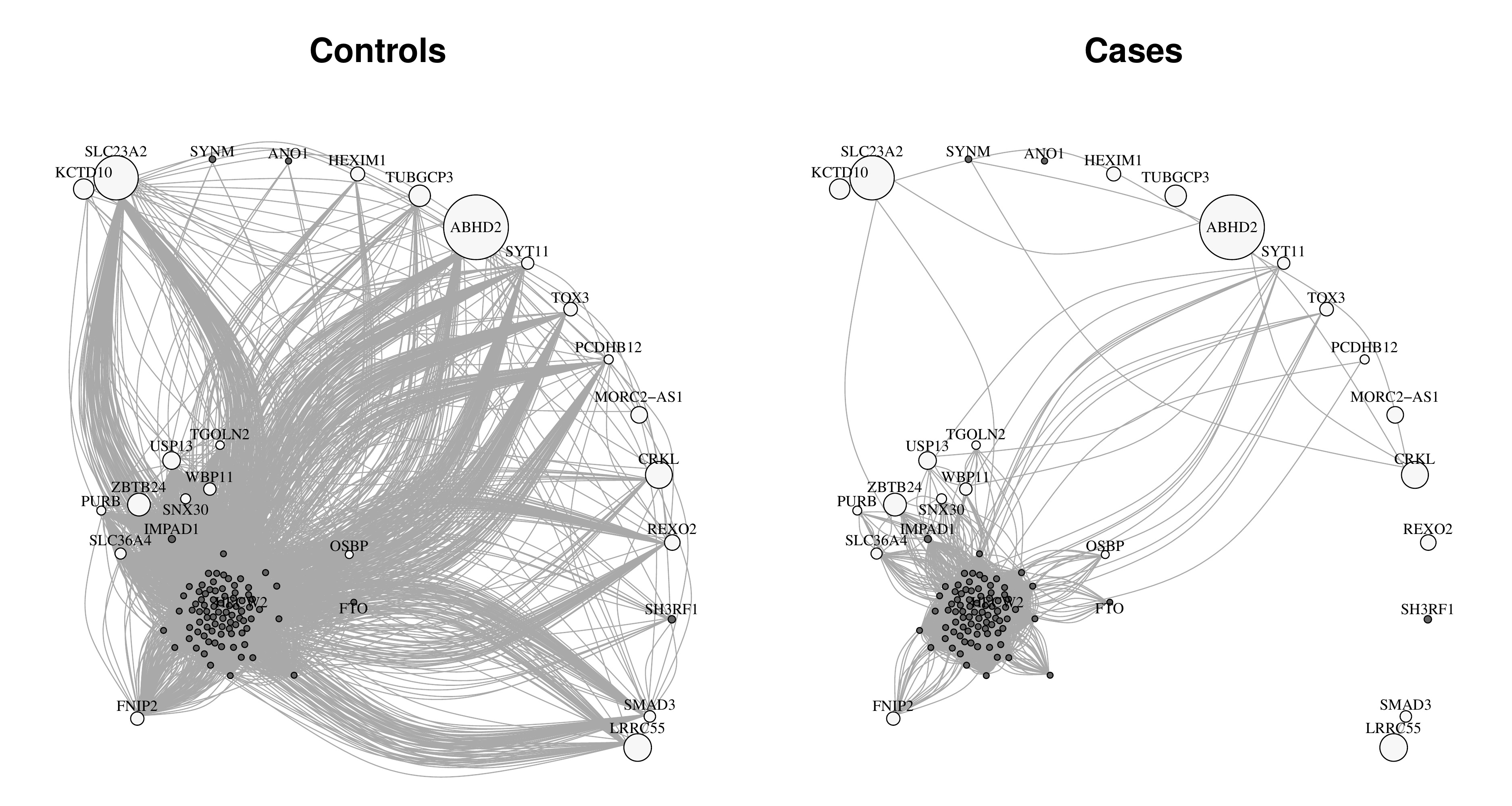}
\caption{Gene networks constructed from control and SCZ samples, using top genes in the {\em M2c} module that have non-zero leverage.
We exclude 6 genes that do not have annotated gene names, and show the remaining 22 primary genes (colored in white) and 85 secondary genes (colored in grey). The adjacency matrix is constructed by thresholding the absolute Pearson correlation $|R_{ij}|$ at 0.5. Larger node sizes represent larger leverage.}
\label{fig:blue-network}
\end{center}
\end{figure}

\bgroup
\def\arraystretch{1.4} % table vertical spacing; 1 is the default
\begin{table}[htbp]
\caption{Annotated names of 22 primary genes and 85 secondary genes in the {\em M2c} module, listed in the descending order of leverage. The 6 underlined genes are also significant in the differential expression analysis}

\begin{center}

\begin{adjustbox}{center}
\begin{tabular}{p{0.7in} p{4.1in}}
\hline
 & {\bf Gene names} \\
\hline
Primary genes & 
ABHD2     SLC23A2   LRRC55    CRKL        ZBTB24       TUBGCP3
KCTD10    USP13     MORC2-AS1 REXO2     HEXIM1    \underline{\it TOX3 }    \underline{\it FNIP2}      WBP11 
SYT11     SMAD3     SLC36A4   SNX30     PCDHB12   PURB      TGOLN2    OSBP
 \\

Secondary genes & 

SH3RF1    IMPAD1    SYNM      HECW2     ANO1       \underline{\it DNM3}         STOX2     C1orf173 
PPM1L     DNAJC6    DLG2      LRRTM4    ANK3      EIF4G3    ANK2      ITSN1
SLIT2     LRRTM3    ATP8A2    CNTNAP2   CKAP5           GNPTAB    USP32 
USP9X    ADAM23    SYNPO2    AKAP11    MAP1B     KIAA1244  PPP1R12B 
SLC24A2   PTPRK     SATB1     CAMTA1    MFSD6     KIAA1279  NTNG1     RYR2     
RASAL2    PUM1      STAM      ST8SIA3   ZKSCAN2   PBX1      ARHGAP24  \underline{\it RASA1} 
ANKRD17   MYCBP2    SLITRK1   BTRC         MYH10     AKAP6     NRCAM
MYO5A     TRPC5     NRXN3     CACNA2D1  DNAJC16   GRIN2A    KCNQ5     NETO1
 FTO       THRB      \underline{\it NLGN1}     HSPA12A   BRAF      OPRM1     KIAA1549L NOVA1 
 OPCML     CEP170    DLGAP1    JPH1      LMO7      PCNX      SYNJ1     RAPGEF2 
 NIPAL2    SYT1      UNC80     ATP8A1    SHROOM2   \underline{\it KCNJ6}     SNAP91    WDR7 
\\
\hline
\end{tabular}
\end{adjustbox}

\end{center}
\label{tab:main-rest}
\end{table}%

To shed light on the nature of the genes identified in the {\it M2c} module, we conduct a Gene Ontology (GO) enrichment analysis \citep{chen:b2013a}. The secondary gene list is most easily interpreted.  It is highly enriched for genes directly involved in synaptic processes, both for GO Biological Process and Molecular Function. Two key molecular functions involve calcium channels/calcium ion transport and glutamate receptor activity. Under Biological Process, these themes are emphasized and synaptic organization emerges too. Synaptic function is a key feature that emerges from genetic findings for SCZ, including calcium channels/calcium ion transport and glutamate receptor activity (see \cite{owen2016schizophrenia} for review).

 
For the primary genes, under GO Biological Process, {\em``regulation of transforming growth factor beta2 (TGF-$\beta$2) production"} is highly enriched. The top GO Molecular Function term is SMAD binding. The protein product of SMAD3 (one of the primary genes) modulates the impact of transcription factor TGF-$\beta$ regarding its regulation of expression of a wide set of genes.  TGF-$\beta$ is important for many developmental processes, including the development and function of synapses \citep{diniz2012astrocyte-induced}. Moreover, and notably, it has recently been shown that SMAD3 plays a crucial role in synaptogenesis and synaptic function via its modulation of how TGF-$\beta$ regulates gene expression \citep{yu2014neuronal}. It is possible that disturbed TGF-$\beta$ signaling could explain co-expression patterns we observe in \cref{fig:blue-network}, because this transcription factor will impact multiple genes. Another primary gene of interest is OSBP.  Its protein product has recently been shown to regulate neural outgrowth and thus synaptic development \citep{gu2015microrna124}. Thus perturbation of a set of genes could explain the pattern seen in \cref{fig:blue-network}.

\subsection{Robustness of the results} 
\textcolor{highlight}{
In this section, we illustrate that sLED remains powerful under perturbation of smoothing parameters and the boundaries of the {\it M2c} module.  We first apply sLED with 1,000 permutations on the {\it M2c} module using $c \in \{0.10, 0.12, \cdots, 0.30\}$ (recall that $\sqrt{R} = c \sqrt{p}$). Each experiment is repeated 10 times, and the average $p$-value and the standard deviation for each experiment are shown in \cref{fig:robust-p}. All of the average $p$-values are smaller than 0.02. Note that larger values of $c$ typically lead to denser solutions, which may hinder interpretability in practice. We also examine the stability of the list of 25 primary genes. For each value of $c$, we record the ranks of these 25 primary genes when ordered by leverage, and their average ranks with the standard deviations are visualized in \cref{fig:robust-rank}. It is clear that these 25 primary genes are consistently the leading ones in all experiments. 
}

\textcolor{highlight}{
Now we examine the robustness of sLED to perturbation of the module boundaries. Specifically, we perturb the {\it M2c} module by removing some genes that are less connected within the module, or including extra genes that are well-connected to the module. 
 As suggested in \cite{zhang2005general}, the selection of genes is based on their correlation with the eigengene (i.e., the first principal component) of the {\it M2c} module, calculated using the 279 control samples. By excluding the {\it M2c} genes with correlation smaller than $\{0.2, 0.3, 0.4\}$, we obtain three sub-modules with sizes $\{1397, 1357, 1248\}$, respectively. By including extra genes outside the {\it M2c} module with correlation larger than $\{0.75, 0.7, 0.65\}$, we obtain three sup-modules with sizes $\{1452, 1537, 1708\}$, respectively. We apply sLED with $c=0.1$ and 1,000 permutations to these 6 perturbed modules. For each perturbed module, we examine the permutation $p$-values in 10 repetitions, as well as the ranks of the 25 primary genes when ordered by leverage. As shown in \cref{fig:robust-p-module} and \cref{fig:robust-rank-module},  the results from sLED remain stable to such module perturbation.
}

\begin{figure}[t!]
\centering

\begin{subfigure}[t]{0.24\textwidth}
        \centering
        \includegraphics[width = \textwidth]{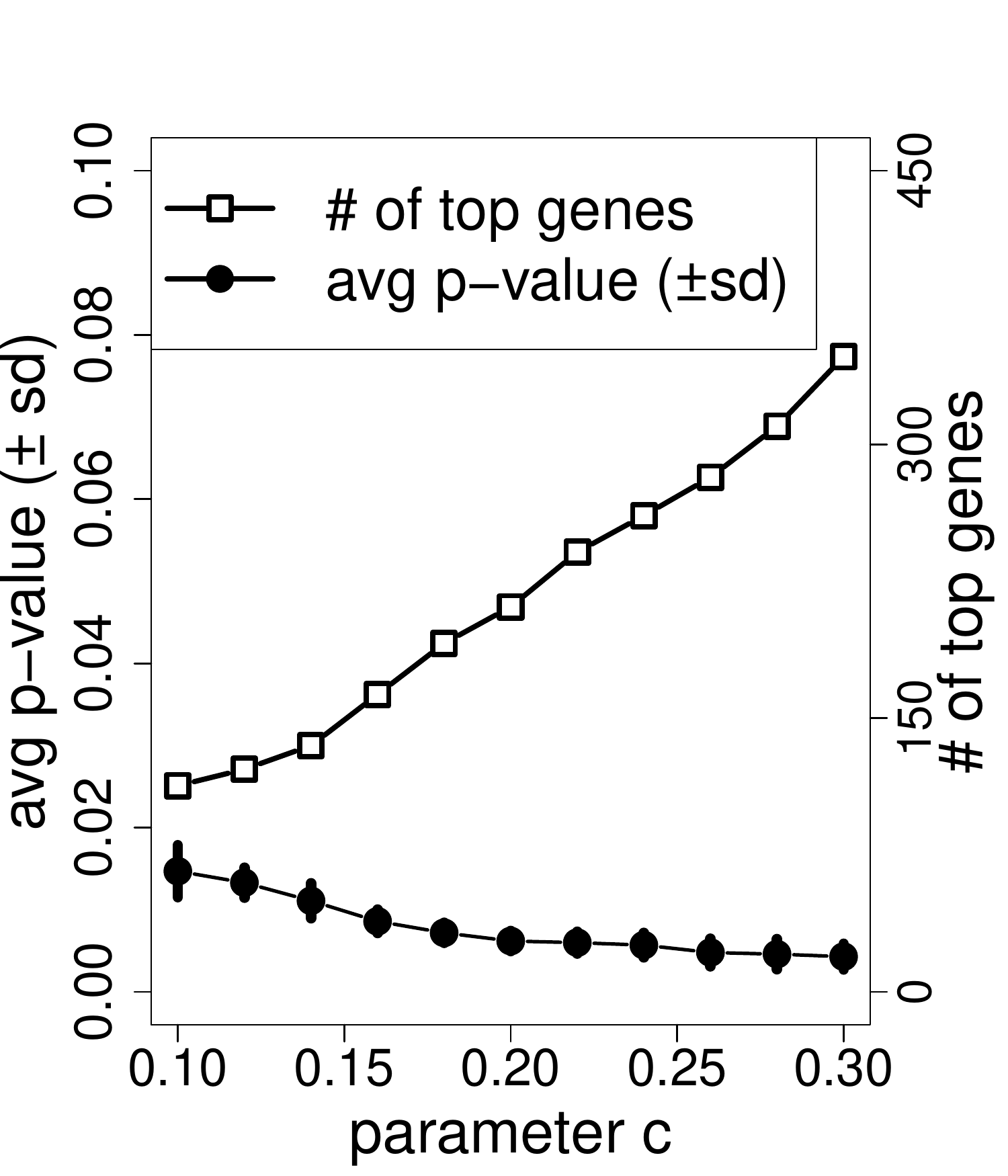}
	\caption{}\label{fig:robust-p}
\end{subfigure}
\begin{subfigure}[t]{0.24\textwidth}
        \centering
        \includegraphics[width = \textwidth]{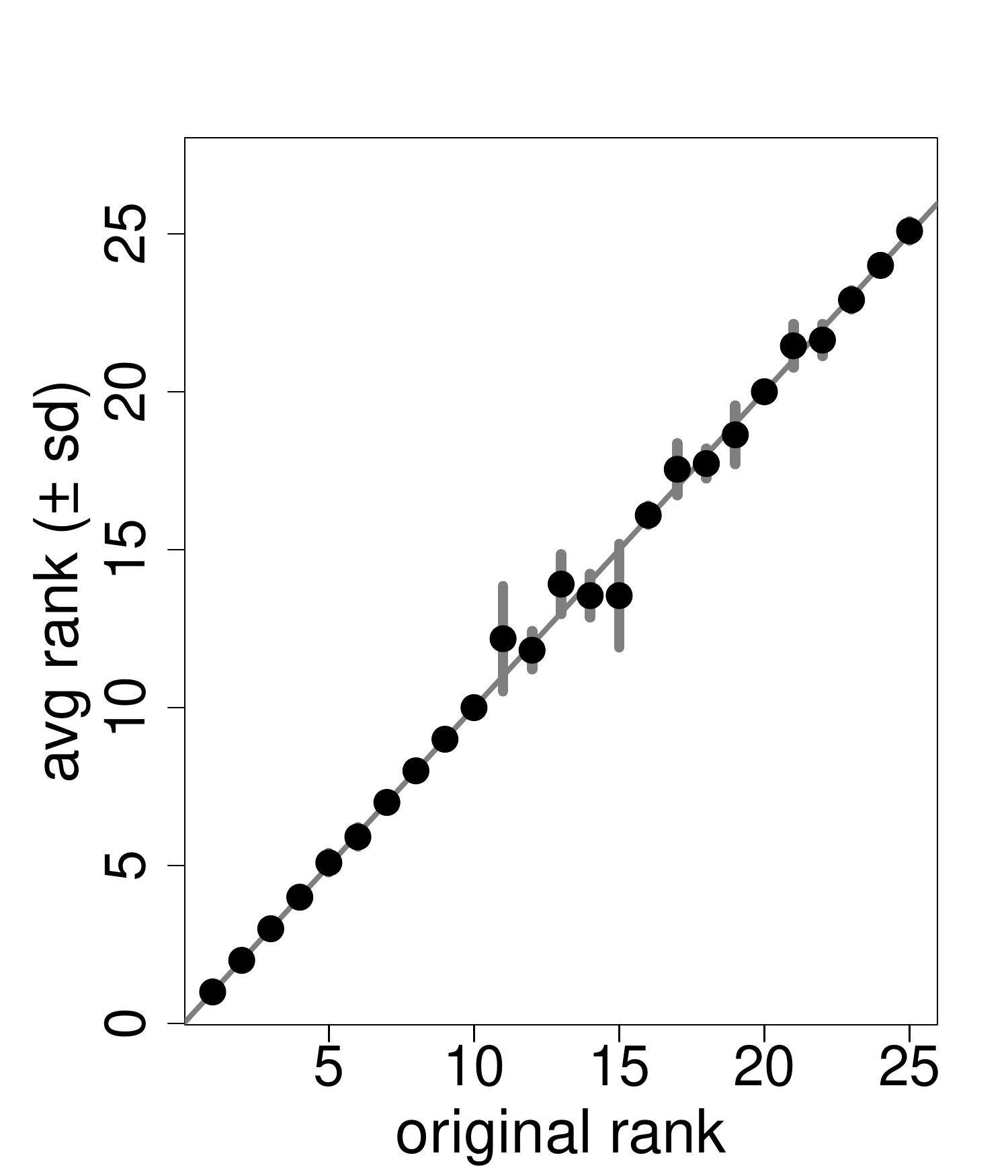}
	\caption{}\label{fig:robust-rank}
\end{subfigure}
~
\begin{subfigure}[t]{0.24\textwidth}
        \centering
        \includegraphics[width = \textwidth]{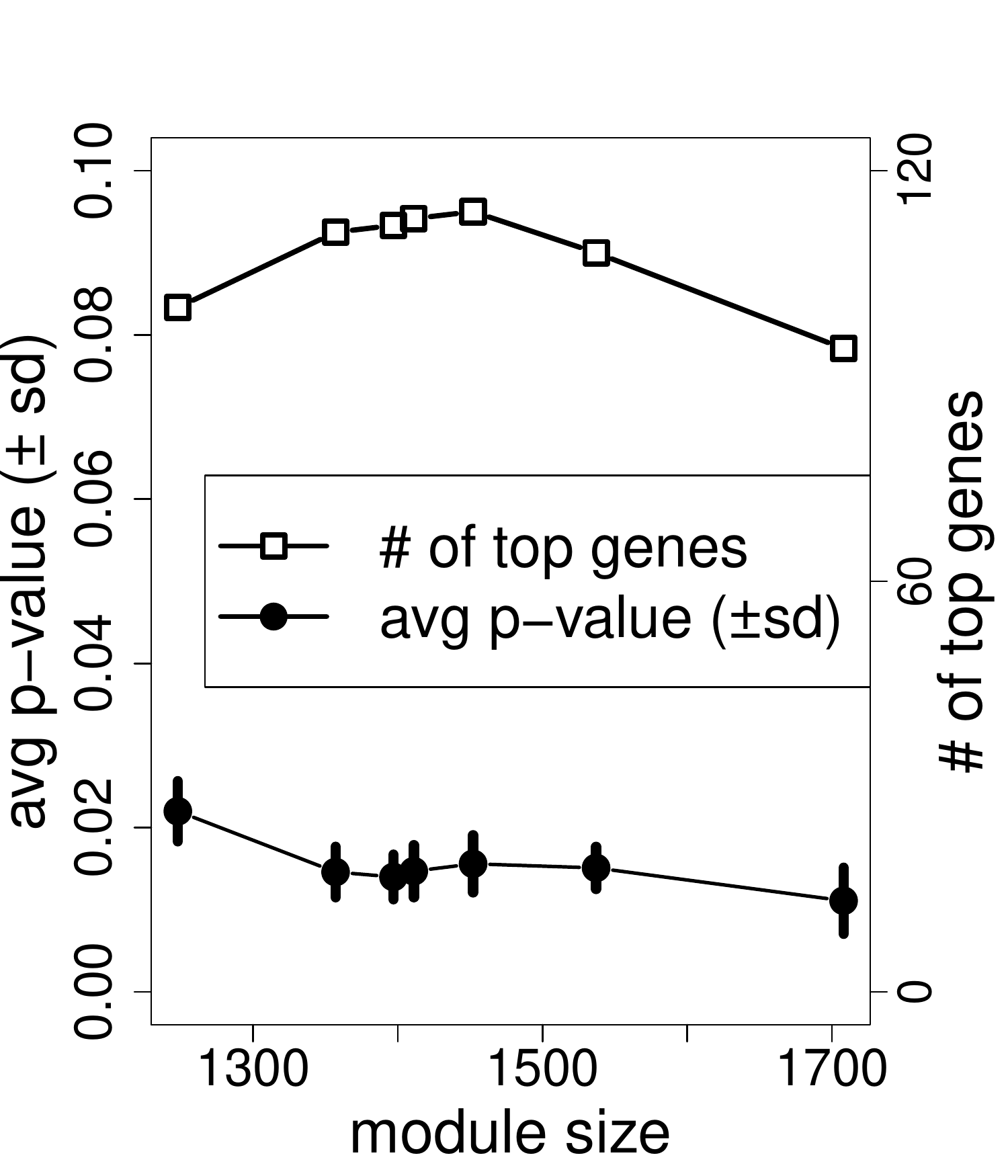}
	\caption{}\label{fig:robust-p-module}
\end{subfigure}
\begin{subfigure}[t]{0.24\textwidth}
        \centering
        \includegraphics[width = \textwidth]{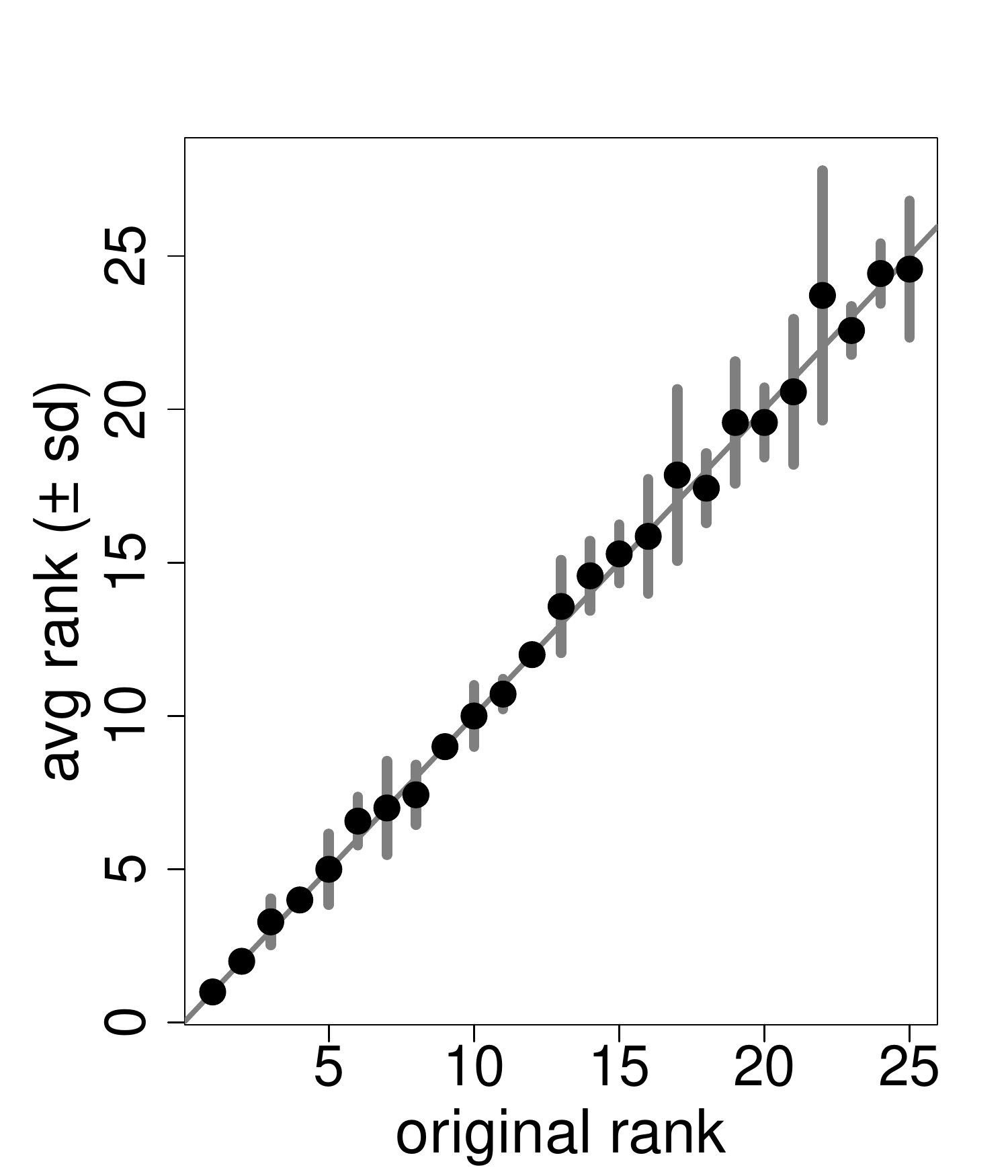}
	\caption{}\label{fig:robust-rank-module}
\end{subfigure}

\caption{
{\em (a)} sLED applied on the {\em M2c} module with $c \in \{0.10, \cdots, 0.30\}$. For each $c$, we visualize the number of genes with non-zero leverage, as well as the average and standard deviation of $p$-values in 10 repetitions, each using 1,000 permutations. 
{\em (b)} The average and standard deviation of ranks of the 25 primary genes when $c \in \{0.10, \cdots, 0.30\}$, where ranks are based on the descending order of leverage.
{\em (c)} sLED applied on 7 modules using $c=0.1$, including the original {\em M2c} module as well as 6 differently perturbed modules with sizes $\{1248, 1357, 1397, 1452, 1537, 1708\}$. For each (perturbed) module, we visualize the number of genes with non-zero leverage, as well as the average and standard deviation of $p$-values in 10 repetitions, each using 1,000 permutations. 
{\em (d)} The average and standard deviation of ranks of the 25 primary genes when sLED is applied on the 7 (perturbed) modules using $c=0.1$, where ranks are based on the descending order of leverage.
}

\label{fig:robust}
\end{figure}

%%% Adj test %%%
\subsection{Generalization to weighted adjacency matrices}
\label{sec:adj-test}
Finally, we illustrate that sLED is not only applicable to testing differences in covariance matrices, but can also be applied to comparing general gene-gene ``relationship'' matrices. As an example, we consider the weighted adjacency matrix $A \in \mathbb{R}^{p \times p}$, defined as
\begin{equation}
 A_{ij} = | R_{ij} |^\beta\,, \textrm{ for } 1 \leq i, j \leq p\,~\textrm{and some constant}~\beta > 0\,, 
 \label{eq:w-adj}
 \end{equation}
where $R_{ij}$ is the Pearson correlation between gene $i$ and gene $j$, and the constant $\beta>0$ controls the density of the corresponding weighted gene network. The weighted adjacency matrix is widely used as a similarity measurement for gene clustering, and has been shown to yield  genetic modules that are biologically more  meaningful than using regular correlation matrices \citep{zhang2005general}.
Now the testing problem becomes
\[ H_0: \tilde{D} = 0\, \ \textrm{versus} \ H_1: \tilde{D} \neq 0\,, \]
where $\tilde{D} = \Exp(A_{control}) - \Exp(A_{SCZ})$.
 While classical two-sample covariance testing procedures are inapplicable under this setting, sLED can be easily generalized to incorporate this scenario. Let $\hat{D} = A_{control} - A_{SCZ}$, then the same permutation procedure as described in \Cref{sec:test-stat} can be applied.

We explore the results of sLED for $\beta \in \{1,\, 3,\, 6.5,\, 9\}$, corresponding  to four different choices of weighted adjacency matrices. We choose the same sparsity parameter $c=0.1$ for sLED as in \cref{sec:data-cov}, and with 1,000 permutations, the $p$-values are $0.020$, $0.001$, $0.002$, and $0.006$ for the 4 choices of $\beta$'s respectively. The latter three are significant at level $0.05$ after a Bonferroni correction. 

\begin{figure}[htbp]
\begin{center}
	
	\begin{subfigure}[t]{0.24\textwidth}
        \centering
        \includegraphics[width = \textwidth]{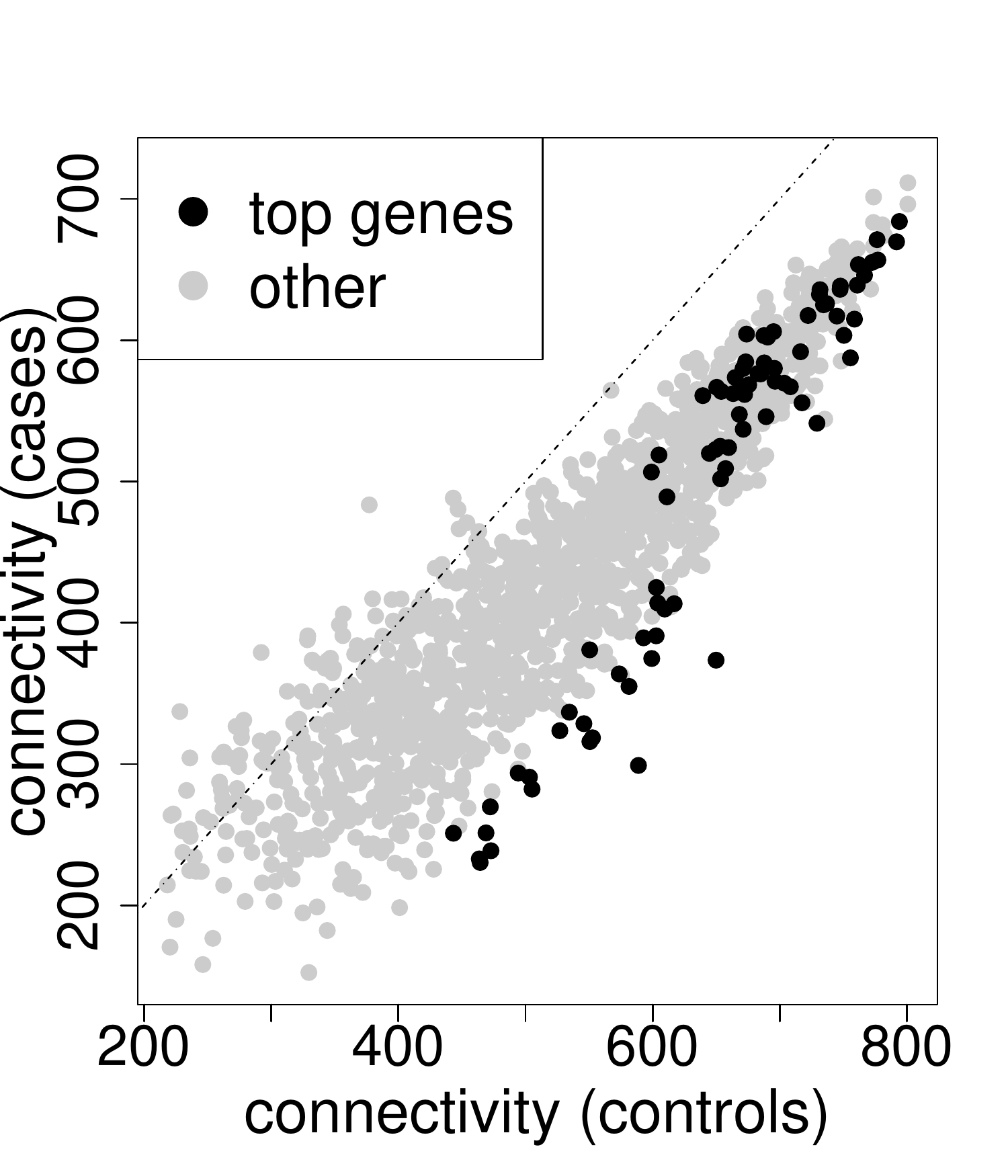}
	\caption{$\beta=1$}
\end{subfigure}
\begin{subfigure}[t]{0.24\textwidth}
        \centering
        \includegraphics[width = \textwidth]{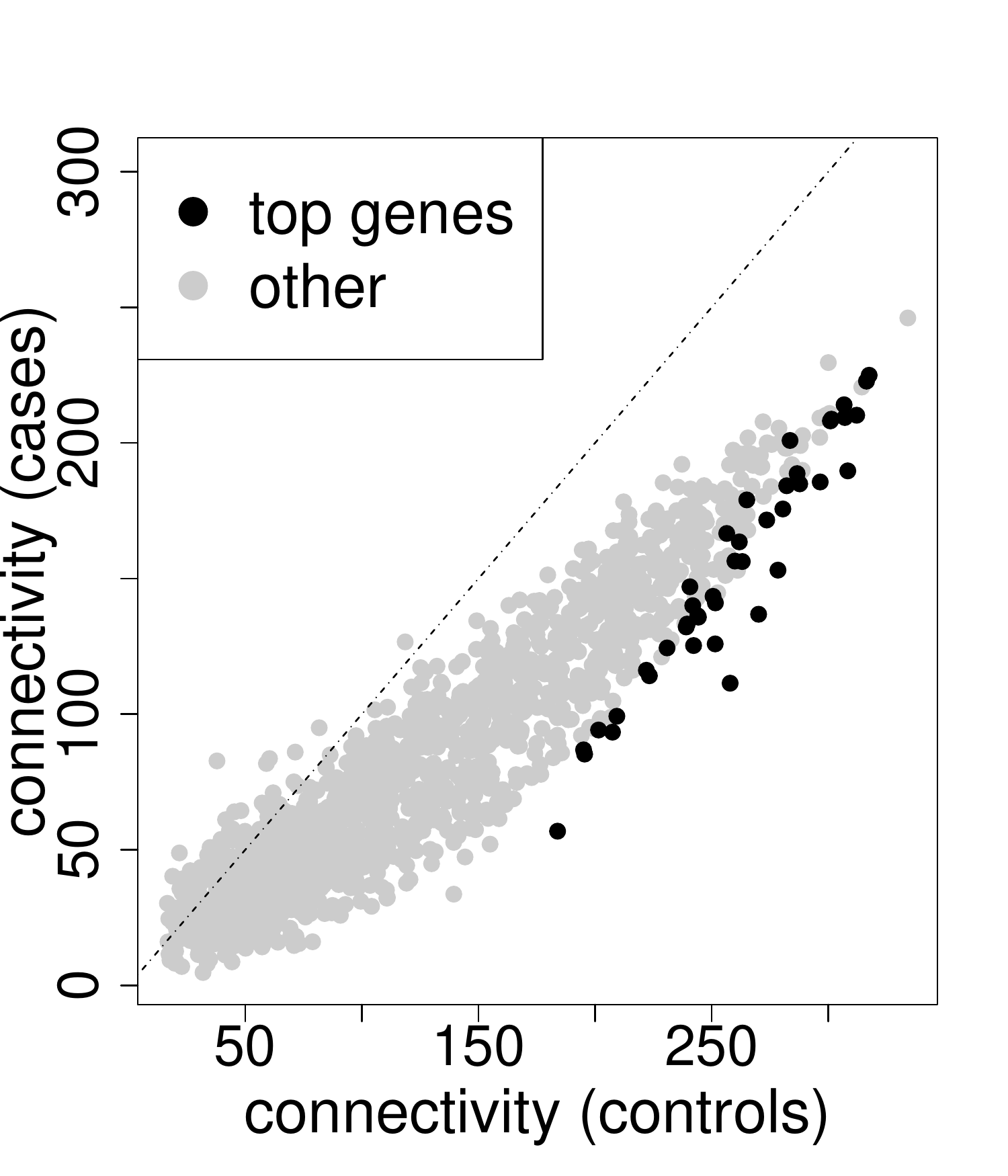}
	\caption{$\beta=3$}
\end{subfigure}
\begin{subfigure}[t]{0.24\textwidth}
        \centering
        \includegraphics[width = \textwidth]{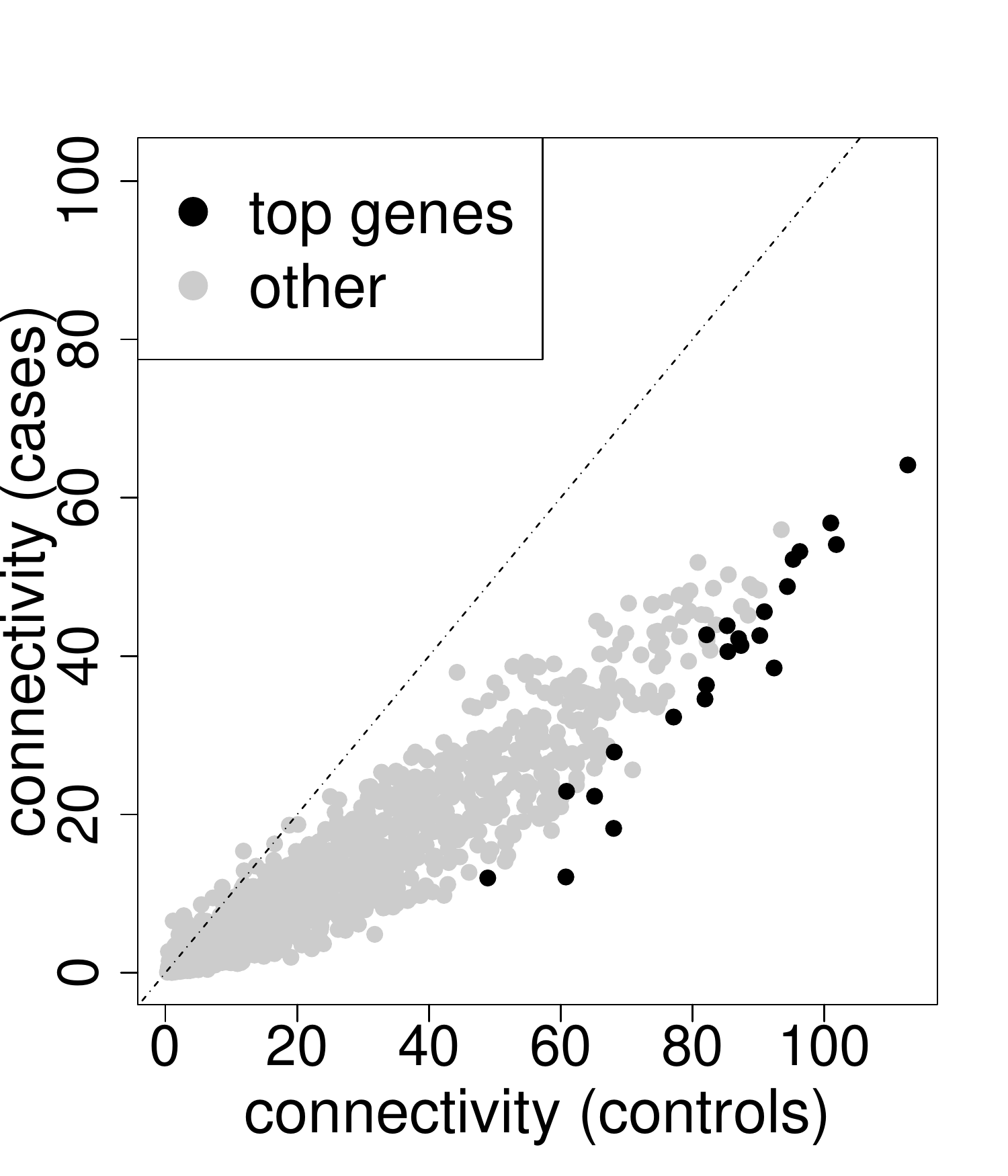}
	\caption{$\beta=6.5$}
\end{subfigure}
\begin{subfigure}[t]{0.24\textwidth}
        \centering
        \includegraphics[width = \textwidth]{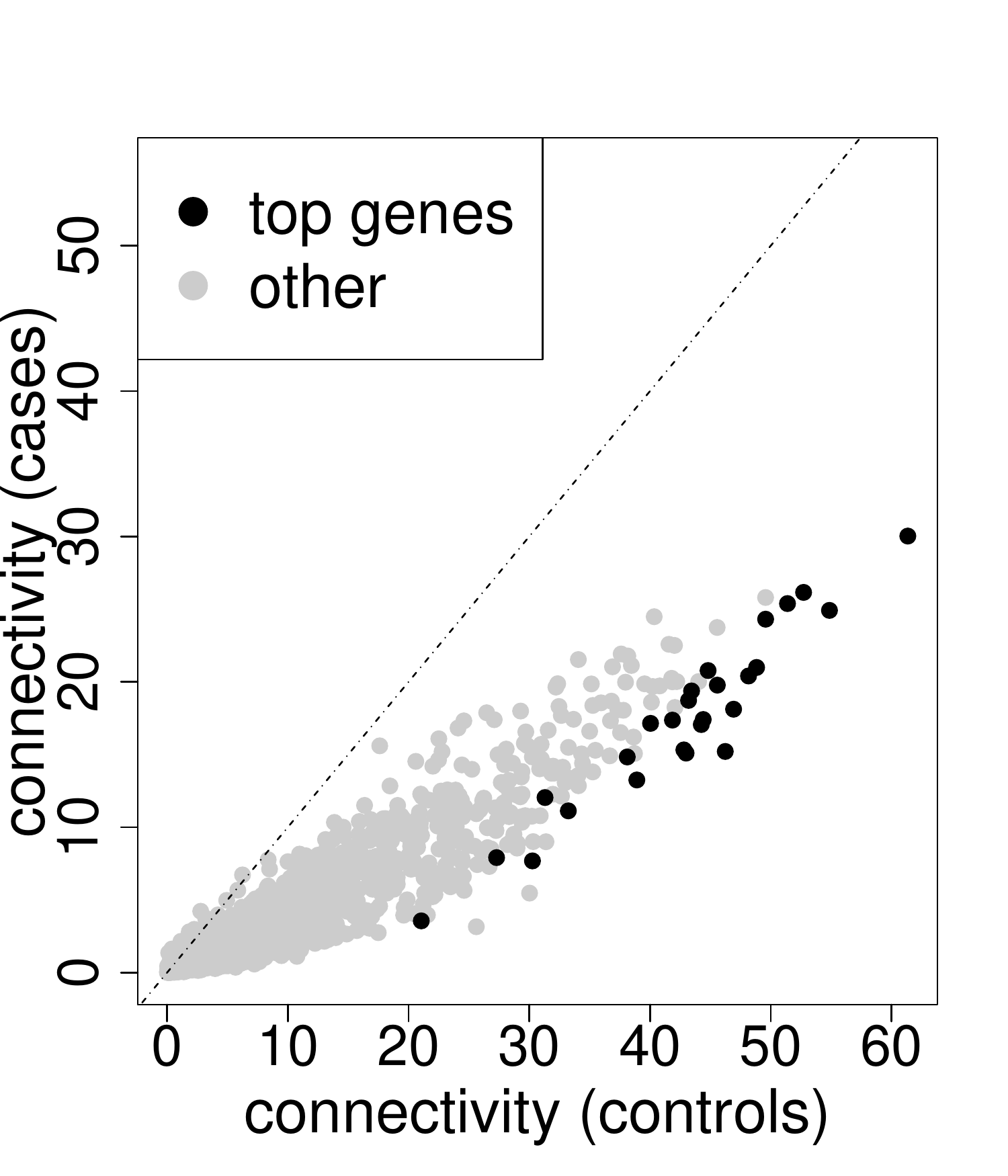}
	\caption{$\beta=9$}
\end{subfigure}	
		
\caption{Connectivity in the {\em M2c} module for control and SCZ samples using weighted adjacency matrices with different $\beta$'s. The top genes detected by sLED with non-zero leverage are highlighted in black, and the auxiliary line $y=x$ is shown in each plot. }
\label{fig:beta-comp}
\end{center}
\end{figure}

Interestingly, we find our results to be closely related to the connectivity of genes in the {\it M2c} module, where the connectivity of gene $i$ is defined as
\[ k_i = \sum_{j \neq i} A_{ij}\,. \] 
\cref{fig:beta-comp} compares the gene connectivities between control and SCZ samples, where the top genes with non-zero leverage detected by sLED are highlighted. It is clear that the connectivity of genes is typically higher in control samples. Furthermore, as $\beta$ increases, the differences on highly connected genes are enlarged, and consistently, the top genes detected by sLED also concentrate more and more on these ``hub'' genes that are densely connected. These genes would have been missed by the covariance matrix test, but are now revealed using weighted adjacency matrices. A Gene Ontology (GO) enrichment analysis \citep{chen:b2013a} highlights a different, although related, set of biological processes when $\beta=9$ versus $\beta=1$ (\Cref{tab:go}). 
%finds that the list of top genes detected when $\beta=9$ reveals more neural-related biological processes than the top genes detected when $\beta=1$ (\Cref{tab:go}).

\bgroup
\def\arraystretch{1.2} % default is 1
\begin{table}[htbp]
\caption{Top 5 terms in Gene Ontology (GO) enrichment analysis on the top genes using weighted adjacency matrices with $\beta \in \{1, 9\}$. The adjusted $p$-values are reported in parentheses.}
\begin{center}

\begin{adjustbox}{center}
\begin{tabular}{p{0.02in}  p{2.65in} | p{2.45in}}
\hline
& $\beta=1$ & $\beta=9$\\
 \hline
1 & Positive regulation of cell development (4.4e-05)  &  Synaptic transmission (5.6e-06) \\
2 &  Axon extension  (4.6e-04)  &  Energy reserve metabolic process (6.5e-06)\\
3 & Regulation of cell morphogenesis involved in differentiation (1.7e-04)      &  Divalent metal ion transport (4.1e-05)   \\
4 &    Neuron projection extension (7.0e-04)    &  Divalent inorganic cation transport  (4.5e-05) \\
5 &  Positive regulation of nervous system development  (6.5e-04)     &  Calcium ion transport  (2.8e-05) \\
\hline
\end{tabular}
\end{adjustbox}

\end{center}
\label{tab:go}
\end{table}%

%!TEX root = ./CovTest-sPCA.tex

\section{Conclusion and discussion}
\label{sec:discussion}

In this paper, we propose sLED, a permutation test for two-sample covariance matrices under the high dimensional regime, which meets the need to understand the changes of gene interactions in complex human diseases. We prove that sLED achieves full power asymptotically; and in many biologically plausible settings, we verify by simulation studies that sLED outperforms many other existing methods. We apply sLED to a recently produced gene expression data set on Schizophrenia, and provide a list of 113 genes that show altered co-expression when brain samples from cases are compared to that from controls. We also reveal an interesting pattern of gene correlation change that has not been previously detected. The biological basis for this pattern is unclear. As more gene expression data become available, it will be interesting to validate these findings in an independent data set.

\textcolor{highlight}{
sLED can be applied to many other data sets for which signals are both sparse and weak. The performance is theoretically guaranteed for sub-Gaussian distributions, but we observe in simulation studies that sLED remains powerful when data has heavier tails. In terms of running time, on the 1,411 genes considered in this paper, sLED with 1,000 permutations takes 40 minutes using a single core on a computer equipped with an AMD Opteron(tm) Processor 6320 @ 2.8 GHz. When dealing with larger datasets, it is straightforward to parallelize the permutation procedure and further reduce the computation time.
}

%We point out that the detection of top genes can have stronger theoretical guaranteehighlights under some extra sparsity assumptions on the differential matrix $D$. In fact, following the proofs in \cite{vu2013fantope, lei2015sparsistency}, an exact support recovery of the principal subspace of $D$ can be achieved. We omit these theoretical results since the required assumptions tend to be too stringent on the CMC data.

Finally, we illustrate that sLED can be applied to a more general class of differential matrices between other gene-gene relationship matrices that are of practical interest. We show an example of comparing two weighted adjacency matrices and how this reveals novel insight on Schizophrenia.
Although we have only stated the consistency results for testing covariance matrices, similar theoretical guarantee may be established for other relationship matrices as long as similar error bounds as in \cref{thm:permutation} hold. 
This is a first step towards testing general high-dimensional matrices, and we leave a more thorough exploration in this direction to future work.

\section*{Acknowledgements}
We thank the editors and anonymous  reviewers for their constructive comments. 
This work was supported by the Simons Foundation SFARI 124827, R37MH057881 (Bernie Devlin and Kathryn Roeder), R01MH103300 (Kathryn Roeder), and National Science Foundation DMS-1407771 (Jing Lei).
Data were generated as part of the CommonMind Consortium supported by funding from Takeda Pharmaceuticals Company Limited, F. Hoffman-La Roche Ltd and NIH grants R01MH085542, R01MH093725, P50MH066392, P50MH080405, R01MH097276, RO1-MH-075916, P50M096891, P50MH084053S1, R37MH057881 and R37MH057881S1, HHSN271201300031C, AG02219, AG05138 and MH06692. Brain tissue for the study was obtained from the following brain bank collections: the Mount Sinai NIH Brain and Tissue Repository, the University of Pennsylvania Alzheimer’s Disease Core Center, the University of Pittsburgh NeuroBioBank and Brain and Tissue Repositories and the NIMH Human Brain Collection Core. CMC Leadership: Pamela Sklar, Joseph Buxbaum (Icahn School of Medicine at Mount Sinai), Bernie Devlin, David Lewis (University of Pittsburgh), Raquel Gur, Chang-Gyu Hahn (University of Pennsylvania), Keisuke Hirai, Hiroyoshi Toyoshiba (Takeda Pharmaceuticals Company Limited), Enrico Domenici, Laurent Essioux (F. Hoffman-La Roche Ltd), Lara Mangravite, Mette Peters (Sage Bionetworks), Thomas Lehner, Barbara Lipska (NIMH).

\bibliography{librarySPCA}
\bibliographystyle{imsart-nameyear}

\pagebreak
\title{Supplemental Materials}
\maketitle

\beginsupplement
\section{Simulations}
\label{sec:simulations-supp}
%!TEX root = CovTest-sPCA-supplement.tex

In this section, we present the remaining simulation results for comparing \texttt{sLED} with other existing methods, including \texttt{Sfrob} \citep{schott2007test}, \texttt{Ustat} \citep{li2012two},  \texttt{Max}
\citep{cai2013two}, \texttt{MBoot} \citep{chang2015bootstrap}, and \texttt{RProj} \citep{wu2015tests}. As explained in \Cref{sec:simulations} of the main paper, we use 100 permutations to compute the $p$-values for all methods, except for \texttt{MBoot} where 100 bootstrap repetitions are used, and we focus on comparing the empirical power.

The samples are generated by $X_i = \Sigma_1^{1/2} Z_i$ for $i = 1, \cdots, n$, and $Y_l = \Sigma_2^{1/2} Z_{n+l}$ for $l = 1, \cdots, m$, where $\{Z_i\}_{i=1, n+m}$ are independent $p$-dimensional random variables with $i.i.d.$ coordinates $Z_{ij}$, $j=1, \cdots, p$. For the different choices of $\Sigma_1$ and $\Sigma_2 = \Sigma_1 + D$, please refer to \Cref{sec:simulations} in the main manuscript. We consider the following four distributions for $Z_{ij}$:
\begin{enumerate}
\item Standard Normal $N(0, 1)$, which leads to multinomial Gaussian samples $X$ and $Y$.
\item Centralized Gamma distribution with $\alpha=4, \beta=0.5$ (i.e., the theoretical expectation $\alpha \beta=2$ is subtracted from $\Gamma(4, 0.5)$ samples). 
\item $t$-distribution with degrees of freedom 12. 
\item Centralized Negative Binomial distribution with mean $\mu=2$ and dispersion parameter $\phi=2$ (i.e., the theoretical expectation $\mu=2$ is subtracted from NB$(2, 2)$ samples). 
\end{enumerate}

\Cref{tab:sim-supp} summarizes the empirical power under different covariance structures and differential matrices when $Z_{ij}$'s are sampled from $t$-distribution and centralized NB$(2,2)$. The results for standard Normal and centralized Gamma distributions are presented in \Cref{tab:sim-power} of the main manuscript. The smoothing parameter for sLED is set to be $\sqrt{R} = 0.3 \sqrt{p}$, and 100 random projections are used for {\tt Rproj}. We also examine the sensitivity of sLED to the smoothing parameter in \cref{fig:sim-sled-r-supp}, where $c$ is varied among $\{0.10, 0.12, \cdots, 0.30\}$ (recall that $\sqrt{R} = c \sqrt{p}$). We see that {\tt sLED} achieves superior power to other approaches under most scenarios, and the results remain robust to many choices of $c$'s.

%%%%%%%%%%%%%%%%%
%% table
%%%%%%%%%%%%%%%%%%
\bgroup
\def\arraystretch{1.3} % table vertical spacing; 1 is the default
\begin{table}[htbp]
\caption{Empirical power in 100 repetitions, where $n=m=100$, nominal level $\alpha = 0.05$, and $Z_{ij}$'s are sampled from  centralized Negative Binomial $(2, 2)$ (top) and t-distribution with degrees of freedom 12 (bottom). Under each scenario, the largest power is highlighted.
}
\label{tab:sim-supp}
\begin{center}

\begin{adjustbox}{center}
\begin{tabular}{l l |  l l l |  l l l |  l l l |  l l l }
\hline
$\mathbf{D}$ &
 \multicolumn{1}{l}{$\mathbf{\Sigma_1}$}
& \multicolumn{3}{c}{\bf Noisy diagonal} & \multicolumn{3}{c}{\bf Block diagonal} & \multicolumn{3}{c}{\bf Exp. decay} & \multicolumn{3}{c}{\bf WGCNA} \\
\cline{3-5}\cline{6-8}\cline{9-11}\cline{12-14}
 & \multicolumn{1}{l}{$\mathbf{p}$} 
 & {\bftab 100} & {\bftab 200} & \multicolumn{1}{c}{\bftab 500}  
 & {\bftab 100} & {\bftab 200} & \multicolumn{1}{c}{\bftab 500}  
 & {\bftab 100} & {\bftab 200} & \multicolumn{1}{c}{\bftab 500}   
 & {\bftab 100} & {\bftab 200} & \multicolumn{1}{c}{\bftab 500}   \\
 \hline
 & \multicolumn{1}{c}{} &  \multicolumn{12}{c}{Centralized Negative Binomial} \\
 {\bf Block}
 & Max & 0.59 & 0.18 & 0.11 & 0.87 & 0.69 & 0.25 & \bftab 1.00 & 0.94 & 0.50 & 0.80 & 0.84 & 0.28 \\ 
 &  MBoot & 0.47 & 0.12 & 0.07 & 0.80 & 0.60 & 0.16 & 0.99 & 0.82 & 0.33 & 0.72 & 0.73 & 0.17 \\ 
 & Ustat & 0.69 & 0.62 & 0.68 & 0.89 & 0.93 & 0.94 & 0.99 & 0.99 & \bftab 1.00 & 0.57 & 0.78 & 0.74 \\ 
 & Sfrob & 0.66 & 0.63 & 0.69 & 0.91 & 0.91 & 0.94 & 0.98 & 0.99 & \bftab 1.00 & 0.58 & 0.79 & 0.80 \\ 
  % & Energy & 0.03 & 0.09 & 0.06 & 0.09 & 0.08 & 0.03 & 0.07 & 0.06 & 0.07 & 0.06 & 0.08 & 0.05 \\ 
 & RProj & 0.11 & 0.13 & 0.06 & 0.18 & 0.13 & 0.11 & 0.21 & 0.21 & 0.10 & 0.15 & 0.10 & 0.16 \\ 
 & sLED & {\bftab 0.91} & \bftab 0.90 & \bftab 0.99 & \bftab 0.99 & \bftab 1.00 & \bftab 1.00 & \bftab 1.00 & \bftab 1.00 & \bftab 1.00 & \bftab 0.86 & \bftab 0.95 & \bftab 0.95 \\ 
 \hline	
 {\bf Spiked}
& Max & 0.04 & 0.07 & 0.06 & 0.60 & 0.20 & 0.10 & 0.93 & 0.82 & 0.25 & 0.93 & 0.41 & 0.06 \\ 
 & MBoot & 0.03 & 0.04 & 0.03 & 0.51 & 0.20 & 0.02 & 0.92 & 0.76 & 0.18 & 0.89 & 0.33 & 0.05 \\ 
 & Ustat & \bftab 0.25 & \bftab 0.12 & 0.03 & 0.82 & 0.35 & 0.13 & 0.98 & 0.94 & 0.72 & 0.36 & 0.12 & 0.03 \\ 
 & Sfrob & \bftab 0.25 & 0.11 & 0.03 & 0.85 & 0.35 & 0.12 & 0.98 & 0.96 & 0.69 & 0.40 & 0.12 & 0.06 \\ 
 %& Energy & 0.05 & 0.05 & 0.01 & 0.10 & 0.06 & 0.06 & 0.09 & 0.04 & 0.05 & 0.06 & 0.02 & 0.01 \\ 
 & RProj & 0.08 & 0.07 & 0.02 & 0.31 & 0.17 & 0.10 & 0.35 & 0.17 & 0.06 & 0.58 & 0.17 & \bftab 0.13 \\
 & sLED & 0.22 & 0.04 & \bftab 0.08 & \bftab 0.97 & \bftab 0.68 &  \bftab 0.16 & \bftab 0.99 & \bftab 1.00 & \bftab 1.00 & \bftab 0.96 & \bftab 0.48 & \bftab 0.13 \\ 
 \hline
 & \multicolumn{1}{c}{} &  \multicolumn{12}{c}{T-distribution} \\
 {\bf Block}
& Max & 0.22 & 0.15 & 0.12 & 0.88 & 0.73 & 0.23 &  \bftab 1.00 & 0.85 & 0.27 & \bftab 0.97 & 0.40 & 0.15 \\ 
 & MBoot & 0.23 & 0.14 & 0.13 & 0.88 & 0.68 & 0.20 &  \bftab 1.00 & 0.84 & 0.23 & 0.91 & 0.37 & 0.13 \\ 
 & Ustat & 0.53 & 0.63 & 0.74 & 0.97 & 0.93 & 0.96 & \bftab 1.00 & 0.99 & \bftab 1.00 & 0.75 & 0.67 & 0.84 \\ 
 & Sfrob & 0.52 & 0.63 & 0.71 & 0.97 & 0.93 & 0.97 & \bftab 1.00 & 0.99 & 0.99 & 0.74 & 0.67 & 0.76 \\ 
  %& Energy & 0.06 & 0.02 & 0.07 & 0.07 & 0.04 & 0.03 & 0.09 & 0.07 & 0.05 & 0.05 & 0.05 & 0.06 \\ 
 & RProj & 0.08 & 0.13 & 0.11 & 0.28 & 0.12 & 0.02 & 0.28 & 0.19 & 0.08 & 0.17 & 0.13 & 0.06 \\ 
 & sLED &  \bftab 0.95 &  \bftab 0.99 &  \bftab 0.99 &  \bftab 0.99 &  \bftab 1.00 &  \bftab 1.00 &  \bftab 1.00 & \bftab 1.00 & \bftab 1.00 & 0.95 & \bftab 0.97 & \bftab 0.92 \\ 
\hline
 {\bf Spiked}
& Max & 0.13 & 0.08 & 0.05 & 0.78 & 0.50 & 0.08 & 0.96 & 0.79 & 0.18 & 0.85 & 0.27 & 0.10 \\ 
&  MBoot & 0.12 & 0.07 & 0.03 & 0.78 & 0.49 & 0.07 & 0.97 & 0.77 & 0.11 & 0.83 & 0.32 &0.13 \\ 
&  Ustat & 0.14 & 0.07 &\bftab  0.10 & 0.79 & 0.36 & 0.07 & \bftab 1.00 & 0.91 & 0.68 & 0.30 & 0.06 & 0.05 \\ 
&  Sfrob & 0.12 & \bftab 0.10 &\bftab  0.10 & 0.80 & 0.36 & 0.07 & \bftab 1.00 & 0.91 & 0.66 & 0.29 & 0.09 & 0.05 \\ 
 % Energy & 0.01 & 0.04 & 0.01 & 0.10 & 0.05 & 0.11 & 0.05 & 0.12 & 0.14 & 0.12 & 0.06 & 0.04 \\ 
 & RProj & 0.07 & 0.06 & 0.06 & 0.34 & 0.20 & 0.08 & 0.36 & 0.16 & 0.07 & 0.57 & 0.27 & \bftab 0.14 \\ 
 & sLED & \bftab 0.40 & 0.09 & 0.03 & \bftab 0.96 & \bftab 0.76 & \bftab 0.14 & \bftab 1.00 & \bftab 1.00 & \bftab 1.00 & \bftab 0.95 & \bftab 0.55 & 0.08 \\ 
  \hline
\end{tabular}
\end{adjustbox}

\end{center}
\end{table}%

\begin{figure}[htbp]	
	\centering
	\begin{subfigure}[b]{\textwidth}
		\centering
		\includegraphics[width=\textwidth]{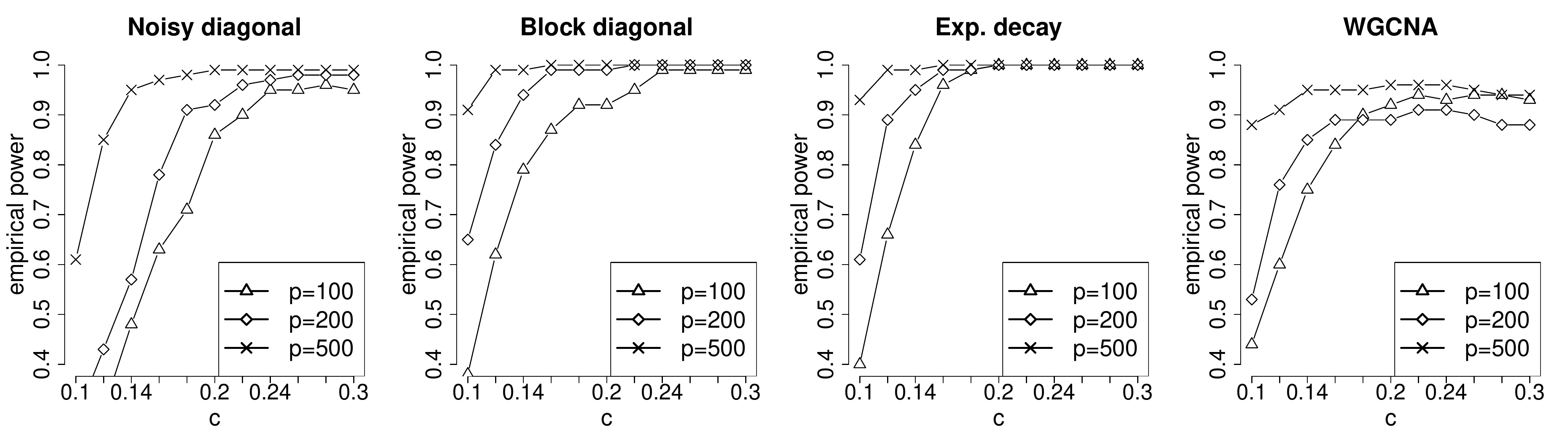}
		\caption{$Z_{ij}\sim$ centralized $\Gamma(4, 0.5)$.}
	\end{subfigure}
	
	\begin{subfigure}[b]{\textwidth}
		\centering
		\includegraphics[width=\textwidth]{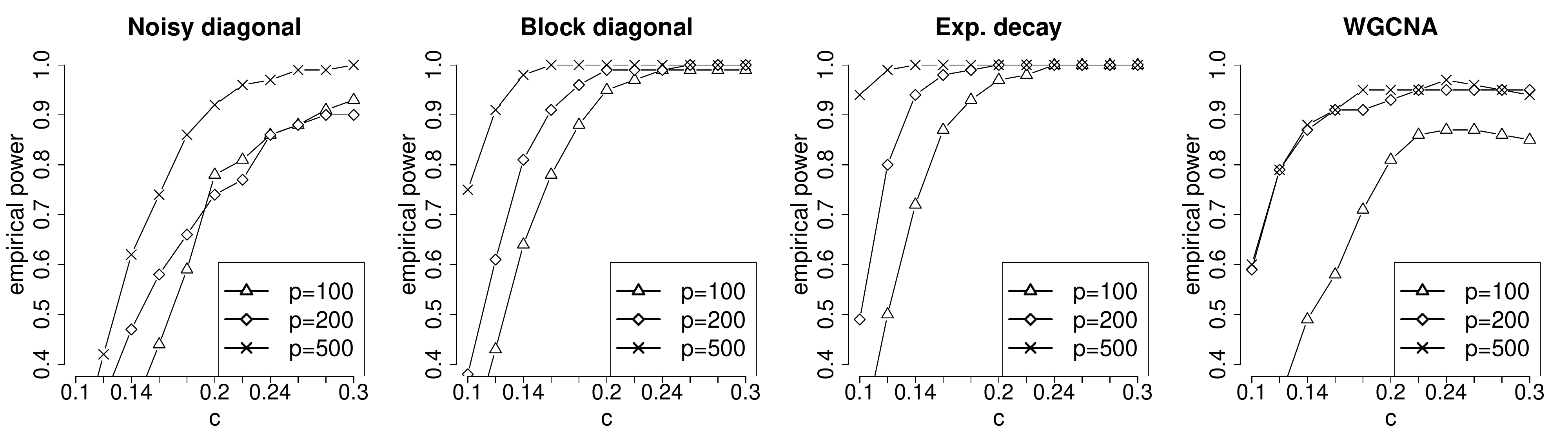}
		\caption{$Z_{ij} \sim$ centralized NB$(2, 2)$.}
	\end{subfigure}
	
	\begin{subfigure}[b]{\textwidth}
		\centering
		\includegraphics[width=\textwidth]{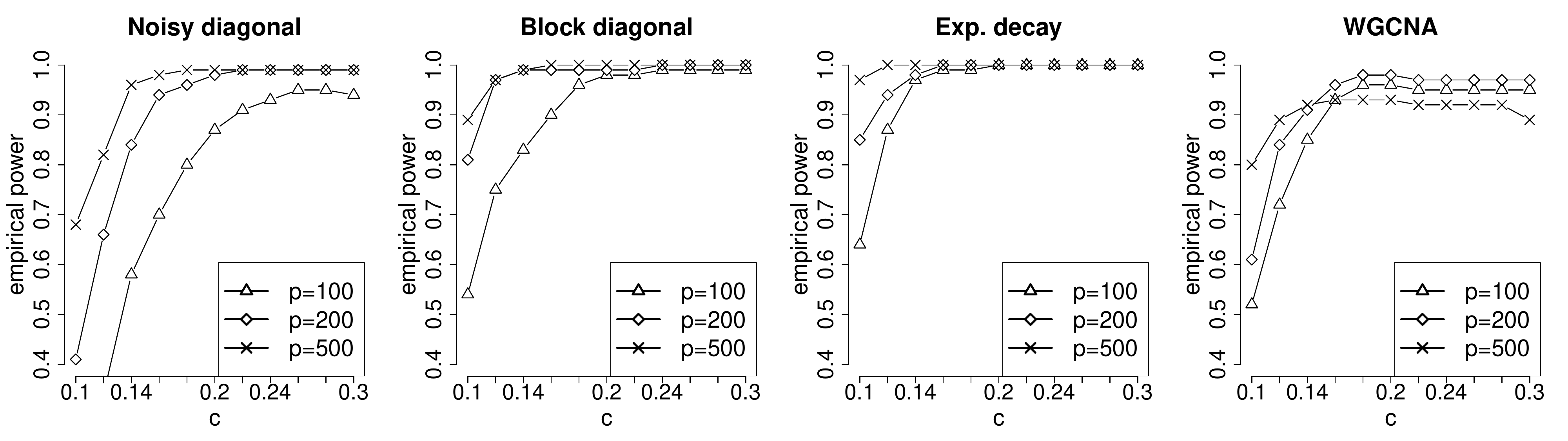}
		\caption{$Z_{ij} \sim t(12)$.}
	\end{subfigure}

	\caption{Empirical power of sLED in 100 repetitions using different smoothing parameters $\sqrt{R} = c \sqrt{p}$ for $c \in \{0.10, \cdots, 0.30\}$, where $D$ has sparse block difference and $Z_{ij}$'s are sampled from different distributions.}
	\label{fig:sim-sled-r-supp}
\end{figure}

\section{Proofs under $L_1$ constraints}
\label{sec:proof-l1}
%!TEX root = CovTest-sPCA-supplement.tex

In this section, we prove \Crefrange{thm:power}{thm:FPS} for the asymptotic power of sLED.

\paragraph{Notation}
For a set $\mathcal{A}$, let $|\mathcal{A}|$ be its cardinality, and $\mathcal{A}^c$ be its complement. 
For $Z = (Z_1, \cdots Z_N) = (X_1, \cdots, X_n, Y_1, \cdots, Y_m)$, we denote $Z_{ki}$ to be the $i$-th coordinate of the $k$-th sample $Z_k$, and
\[\begin{split}
& \hat{\Sigma} = \frac{1}{N} \sum_{k=1}^N Z_k Z_k^T\,,\ \bar{Z} =  \frac{1}{N} \sum_{k=1}^N Z_k = (\bar{Z}_1, \cdots, \bar{Z}_p)^T\,, \\
& m_z = || Z ||_\infty\,, \ \overline{m}_z = || \bar{Z} ||_\infty\,, \  \overline{m}_z^{(2q)} = \max_{1 \leq i,j \leq p} \frac{1}{N} \sum_{k=1}^N Z_{ki}^q Z_{kj}^q\,, \,  q=1, 2\,. 
\end{split}\]

%%%%%%%%%%%%
%% Theorem 1
%%%%%%%%%%%%
\begin{proof}[Proof of \cref{thm:power}]

By \Cref{thm:permutation} and \Cref{lemma:sub-gaussian}, there exist some constants $C', C''$ depending on $(\underline{c}, \bar{c}, \nu^2, \delta)$, such that if $(n, p)$ are sufficiently large,  with probability at least $1 - \delta$, 
\[ || \hat{D}^* ||_\infty \leq C' \sqrt{ \frac{\log p}{ n }}\,, \  ||\hat{D} - D||_\infty \leq C'' \sqrt{ \frac{\log p}{n} }\,. \]
Then we apply \Cref{thm:FPS} on both $\hat{D}, - \hat{D}$ and $\hat{D}^*, - \hat{D}^*$, and this together with assumption (A4) imply the desired conclusion with $C = C' + C''$.
\end{proof}

%%%%%%%%%%%%
%% Theorem 2
%%%%%%%%%%%%
\begin{proof}[Proof of \cref{thm:permutation}]
First, note that for $\forall \epsilon > 0$,
\begin{equation}
 \Prob \left( || \hat{D}^* ||_{\infty} > \epsilon \right) \leq 
\Prob \left( || \hat{\Sigma}_1^* - \hat{\Sigma} ||_{\infty} > \frac{\epsilon}{2}  \right) +
\Prob \left( || \hat{\Sigma}_2^* - \hat{\Sigma} ||_{\infty} > \frac{\epsilon}{2}  \right)\,. 
\label{eq:triangle}
\end{equation}
Now for any $\delta > 0$ and constants $C_1, C_2$,  define 
\[\begin{split} \mathcal{A} = \left\{ Z: \,
m_z \leq C_2 \sqrt{ \log \left( \frac{C_1 np}{ \delta} \right) } \,,\,
\overline{m}_z \leq C_2 \sqrt{ \frac{\log (C_1 p / \delta) }{n} } \,,\,
 \overline{m}_z^{(2q)} \leq C_2 \,,\, q=1, 2
%\overline{m}_z^{(4)} \leq C_2 \left[ 1 +\frac{\log (C_1 p^2 / \delta)}{n} \right]
 \right\}. 
 \end{split}\]
By \Cref{lemma:sub-gaussian}, there exist some constants $C_1$, $C_2$ depending on $(\underline{c}, \bar{c}, \nu^2)$, such that if $(n, p)$ are sufficiently large,
$ \Prob \left( Z \not\in \mathcal{A} \right) \leq \delta/4$.
Therefore, in order to show that
\[ \Prob \left( || \hat{\Sigma}_1^* - \hat{\Sigma} ||_{\infty} > \frac{\epsilon}{2}  \right) \leq 
\Prob \left( || \hat{\Sigma}_1^* - \hat{\Sigma} ||_{\infty} > \frac{\epsilon}{2} \middle| Z \in \mathcal{A} \right) + 
\Prob \left( Z \not\in \mathcal{A} \right) \leq \frac{\delta}{2}\,, \]
it suffices to show that given any $Z \in \mathcal{A}$, the conditional probability satisfies
\begin{equation}
\Prob_Z \left( || \hat{\Sigma}_1^* - \hat{\Sigma} ||_{\infty} > \frac{\epsilon}{2} \right) 
\leq \frac{\delta}{4} \,.
\label{eq:sigma-bound}
\end{equation}
For any $1 \leq i, j \leq p$, we first bound the $(i, j)$-th entry:
\[\begin{split}
 \Prob_{Z} \left( | \hat{\Sigma}_{1, ij}^* - \hat{\Sigma}_{ij} |  > \frac{\epsilon}{2}  \right) \leq  &
\underbrace{\Prob_Z \left(  \left|  \frac{1}{n} \sum_{k=1}^n Z_{ki}^* Z_{kj}^* - \frac{1}{N} \sum_{k=1}^N Z_{ki} Z_{kj} \right| > \frac{\epsilon}{4} \right)}_{\Delta_1}  + \\
& + \underbrace{\Prob_Z \left(  \left|  \bar{X}_i^* \bar{X}_j^* - \bar{Z}_i \bar{Z}_j  \right| > \frac{\epsilon}{4} \right)}_{\Delta_2} \,,
\end{split}\]
where $\bar{X}_i^* = \frac{1}{n} \sum_{k=1}^n Z_{ki}^*$. Now we bound $\Delta_1$ and $\Delta_2$ separately.

\begin{enumerate}[(i)]
\item $\Delta_1$: Note that for any $(k, i, j)$,
 \[ \left| Z_{ki}^* Z_{kj}^* \right| \leq (m_z)^2\,, \
 \textrm{var}_Z \left(Z_{ki}^* Z_{kj}^*\right) \leq \frac{1}{N} \sum_{l=1}^N Z_{li}^2 Z_{lj}^2 \leq \overline{m}_z^{(4)}\,. 
 \]
By \Cref{lemma:bernstein}, there exists a constant $C_2'$ depending on $(C_2, \nu^2)$, such that if $(n, p)$ are sufficiently large,
\begin{equation}
 \Delta_1 
 %\leq 2 \exp \left\{ - \frac{ n \epsilon^2 }{32 m_z^{(4)} + \frac{16}{3} m_z^2 \epsilon} \right\}
\leq 2 \exp \left\{ - \frac{ n \epsilon^2 / C_2' }{ 1+ \log (C_1 np / \delta)  \epsilon} \right\}. 
\label{eq:delta1}
\end{equation}

\item $\Delta_2$: Note that
\[ \bar{X}_i^* \bar{X}_j^* - \bar{Z}_i \bar{Z}_j = (\bar{X}_i^* - \bar{Z}_i)(\bar{X}_j^* - \bar{Z}_j) + \bar{Z}_j (\bar{X}_i^* - \bar{Z}_i) +  \bar{Z}_i (\bar{X}_j^* - \bar{Z}_j)\,, \]
and for any $(k, i, j)$,
\[ | \bar{Z}_i | \leq \overline{m}_z\,, \  
| Z_{ki}^* | \leq m_z\,, \ 
\textrm{var}_Z (Z_{ki}^*) \leq \frac{1}{N} \sum_{l=1}^N Z_{li}^2 \leq \overline{m}_z^{(2)}\,. \]
Therefore, 
\[ \Delta_2 \leq 2 \max_i \ \left[ \Prob_Z \left( \left| \frac{1}{n} \sum_{k=1}^n Z_{ki}^* - \bar{Z}_i \right| > \sqrt{\frac{\epsilon}{8}} \right) + \Prob_Z \left( \left| \frac{1}{n} \sum_{k=1}^n Z_{ki}^* - \bar{Z}_i \right| > \frac{\epsilon}{16 \overline{m}_z} \right) \right]. \]
Applying \Cref{lemma:bernstein} on both terms, we know that there exists a constant $C_2''$ depending on $(C_2, \nu^2)$, such that if $(n, p)$ are sufficiently large,
\begin{equation}\begin{split}
 \Delta_2 
 %\leq & 4 \left[  
%\exp \left\{ - \frac{n \epsilon}{16 m_z^{(2)} + \frac{8 \sqrt{2}}{3} m_z \sqrt{\epsilon}} \right\}  +
%\exp \left\{ - \frac{n \epsilon^2}{512 \overline{m}_z^2 \, m_z^{(2)} + \frac{64}{3} m_z \overline{m}_z \epsilon } \right\}
%\right] \\
\leq & 4 
\exp \left\{ - \frac{n \epsilon / C_2''}{1 + \sqrt{ \log (C_1 np / \delta) } \sqrt{\epsilon}} \right\} + \\
& \qquad 4\exp \left\{ - \frac{n \epsilon^2 / C_2''}{ \frac{\log (C_1 p / \delta)}{n} + \sqrt{ \frac{  \log (C_1 p / \delta) \log (C_1 np / \delta) }{ n } } \epsilon } \right\}\,.
\label{eq:delta2}
\end{split}\end{equation}
\end{enumerate}

Combining the results in \cref{eq:delta1} and \cref{eq:delta2}, and note that $(\log p)^3 = O(n)$ by assumption (A3), we have $\Delta_1\,, \Delta_2 \leq \frac{\delta}{8} p^{-2}$ if $(n, p)$ are sufficiently large, as long as
\[ \epsilon \geq C' \sqrt{ \frac{\log (C_1 p^2 / \delta)}{n} }  \]
for some constant $C'$ depending on $C_2'$ and $C_2''$. Finally, \cref{eq:sigma-bound} follows from a union bound over $1 \leq i, j \leq p$. Similar statement also holds for $|| \hat{\Sigma}_2^* - \hat{\Sigma}||_{\infty}$ with sample size $m$, and the final result follows from \cref{eq:triangle} and the fact that $\underline{c}n \leq m \leq \bar{c}n$. 
\end{proof}

%%%%%%%%%%%%
%% Theorem 3
%%%%%%%%%%%%
\begin{proof}[Proof of \cref{thm:FPS}]
\begin{enumerate}[(i)]
\item Note that a feasible solution of \cref{eq:fps} or \cref{eq:aug-pmd} always satisfies $||H||_1 \leq R$, where $H =  v v^T$ if using \cref{eq:aug-pmd}. Then the result directly follows from the H\"older's inequality:
\[ \textrm{tr} \left(\hat{D} H \right) \leq || \hat{D} ||_{\infty} || H ||_1\,. \]

\item Let $v^*$ be the $R$-sparse leading eigenvector of $D$, then $||v^*||_2=1$ and $||v^* (v^*)^T||_1 = || v^* ||_1^2 \leq ||v^*||_0 = R$, so $v^* (v^*)^T$ is feasible for \cref{eq:fps} and \cref{eq:aug-pmd}. The result follows from
\[ \tilde{\lambda}_1^R(\hat{D}) - \lambda_1^R(D) \geq (v^*)^T \hat{D} v^* -  (v^*)^T D v^*  \]
and
$\left| (v^*)^T (\hat{D} - D ) v^* \right|  \leq ||\hat{D} - D ||_\infty || v^* (v^*)^T ||_1
$.
\end{enumerate}
\end{proof}

\section{Proofs under $L_0$ constraints}
\label{sec:proof-l0}
%!TEX root = covTest-SPCA-supplement.tex

In this section, we prove \Cref{thm:power-l0} for the power of sLED test under $L_0$-sparsity.  We use the same notation as introduced in the beginning of \Cref{sec:proof-l1}. The proof of  \Cref{thm:power-l0} is built on the following two theorems.

%%%%%%%%%%%%%%
%% Theorem 5
%%%%%%%%%%%%%%
\begin{theorem}[Permutation test statistic under $L_0$ constraint] 
\label{thm:perm-l0}
Let $\hat{D}^*$ be the permutation differential matrix as defined in \cref{eq:permutation-cov}, and $\lambda_1^R (\hat{D}^*)$ be the exact solution of \cref{eq:sparse-eigen}. Then under assumptions (A1)-(A2), for any $\delta > 0$, there exist constants $C_1, C_2, C_3$ depending on $(\underline{c}, \bar{c}, \nu^2)$, such that with probability at least $1 - \delta$,
\[ \lambda_1^R (\hat{D}^*) \leq h(C_1, C_2, C_3, t)\,,  \]
where 
\[h(C_1, C_2, C_3,  t) = C_1 s \left[ \log (C_3 Np) +t \right] \frac{t}{N} + C_2  \sqrt{  \left(1 +  \frac{t+ s\log (9ep/s) }{ N} \right) \frac{t}{N}}\,,\]
and $s = \lfloor R \rfloor$,  $t = s \log \left( 9ep / s \right) + \log (1 / \delta)$. 
\end{theorem}

\begin{proof}
Following \cite{vershynin2010introduction}, for an integer $s$, there exists a $\frac{1}{4}$-net $\mathcal{N}_s$ over the unit sphere $\mathbb{S}^{s-1}$, such that $| \mathcal{N}_s | \leq 9^s$, and for any matrix $A \in \mathbb{R}^{s \times s}$, 
\[ \lambda_1(A) \leq 2 \max_{v \in \mathcal{N}_s} v^T A v\,. \]
Therefore, for any subset $\mathcal{S} \subseteq \{1, ..., p\}$, let $\hat{D}_{\mathcal{S}}^*$ be the sub-matrix on $\mathcal{S} \times \mathcal{S}$, we have
\[ \lambda_1^R (\hat{D}^*) = \lambda_1^s (\hat{D}^*) =
 %\max_{|\mathcal{S}| = s} \max_{\substack{||u||_2 = 1 \\ \textrm{supp}(u) \subseteq \mathcal{S}}} u^T   \hat{D}^* u =
  \max_{|\mathcal{S}| = s} \lambda_1 \left(\hat{D}_{\mathcal{S}}^*\right)
\leq 2 \max_{|\mathcal{S}| = s} \, \max_{v \in \mathcal{N}_s} v^T \left( \hat{D}_{\mathcal{S}}^* \right) v\,.
 \]
Moreover, for any given $v \in \mathcal{N}_s$ and subset $\mathcal{S}$, we can construct $u \in \mathbb{S}^{p-1}$ that is augmented from $v \in \mathbb{S}^{s-1}$ by adding zeros on coordinates in $\mathcal{S}^c$, then 
\[ v^T \left( \hat{D}_{\mathcal{S}}^* \right) v = u^T \hat{D}^* u\,. \]
We define the collection of such $u$'s to be
\[ \widetilde{\mathcal{N}}_s = \{ u \in \mathbb{R }^p: ||u||_2=1, \textrm{supp}(u) \subseteq \mathcal{S},  |\mathcal{S}| = s, u(\mathcal{S})  \in \mathcal{N}_k \}\,, \]
where $u(\mathcal{S})$ is the sub-vector restricted on coordinates in $\mathcal{S}$, and we have 
\[ \left| \widetilde{\mathcal{N}}_s \right| = {p \choose s} \left| \mathcal{N}_s \right| \leq \left(\frac{9ep}{s} \right)^s\,. \]
Next, we show that there exist constants $C_1, C_2, C_3$ depending on $(\underline{c}, \bar{c}, \nu^2)$, such that
 \begin{equation}
  \Prob \left( u^T \hat{D}^*u \geq h \left(\frac{C_1}{2}, \frac{C_2}{2}, C_3, t \right) \right)  \leq e^{-t}, \, \forall t > 0, \, \forall u \in \widetilde{N}_s\,.
  \label{eq:single-bound}
  \end{equation}
 Note that
 \[  u^T \hat{D}^*u  = \frac{1}{m} \sum_{l=(n+1)}^{N} \left( u^T Z_l^* \right)^2 - \frac{1}{n} \sum_{k=1}^n \left( u^T Z_k^* \right)^2\,, \]
and we define
\[ \gamma_z = \max_{u \in \widetilde{\mathcal{N}}_s} \max_{1 \leq k \leq N} Z_k^T u\,, \
  \bar{\gamma}_z^{(4)} = \max_{u \in \widetilde{\mathcal{N}_s}} \frac{1}{N} \sum_{k=1}^N (Z_k^T u)^4\,,
 \]
 and 
 \[\mathcal{G} = \left\{ Z: \gamma_z \leq C_1' \sqrt{s } \sqrt{\log (C_3' Np) + t}\,, \ 
  \bar{\gamma}_z^{(4)} \leq C_2' \left(1 +  \frac{t+ s\log (9ep/s) }{ N} \right) \right\}.\]
Now for any $t > 0$, by \Cref{lemma:sub-tail},  there exist constants $C_1', C_2', C_3'$ depending on $\nu^2$, such that
 \[ \Prob \left( Z \in \mathcal{G} \right) \geq 1 - \frac{e^{-t}}{2}\,. \]
Therefore, in order to prove \cref{eq:single-bound}, it suffices to show that given any $Z \in \mathcal{G}$, the conditional probability satisfies
 \[ \Prob_Z \left( \frac{1}{m} \sum_{l=(n+1)}^{N} \left( u^T Z_l^* \right)^2 - \frac{1}{n} \sum_{k=1}^n \left( u^T Z_k^* \right)^2 \geq h \left(\frac{C_1}{2}, \frac{C_2}{2}, C_3, t \right) \right) \leq \frac{e^{-t}}{2}\,. \]
Note that given $Z \in \mathcal{G}$, $\left( u^T Z_k^* \right)^2$ satisfies
\[\begin{split}
& \max_{1 \leq k \leq N} \left( u^T Z_k^* \right)^2 \leq (\gamma_z)^2,\  \textrm{var}_Z \left[  \left( u^T Z_k^* \right)^2 \right]  \leq  \bar{\gamma}_z^{(4)}\,, \forall k=1, \cdots, N\,.
\end{split}\] 
Therefore, by \Cref{lemma:bernstein}, there exist constants $C_1, C_2, C_3$ depending on $(\nu^2, C_1', C_2', C_3')$, such that
\[ \Prob_Z \left( \left| \frac{1}{n} \sum_{k=1}^n \left( u^T Z_k^* \right)^2 - \frac{1}{N} \sum_{k=1}^N \left( u^T Z_k \right)^2 \right| \geq h \left( \frac{C_1}{2}, \frac{C_2}{2}, C_3, t \right) \right) \leq \frac{e^{-t}}{4}\,. \]
Similar results also hold for $ \frac{1}{m} \sum_{l=(n+1)}^{N} \left( u^T Z_l^* \right)^2$, and \cref{eq:single-bound} follows from $\underline{c} n \leq m \leq \bar{n}$.
%\[ \Prob_Z \left( \left| \frac{1}{n} \sum_{k=1}^n \left( u^T Z_k^* \right)^2 - \frac{1}{N} \sum_{k=1}^N \left( u^T Z_k \right) \right| \geq \epsilon \right) \leq 2 \exp \left\{ - \frac{n \epsilon^2}{ 2 C_2' + \frac{4}{3} C_1'^2 s \left( \log (C_3 Np) + t \right) \epsilon } \right\}\,. \]
Finally, with a union bound over $\widetilde{\mathcal{N}}_s$, we have, for any $\delta > 0$,
\[\begin{split}
 \Prob \left( \lambda_1^R (\hat{D}^*) > h(C_1, C_2, C_3, t) \right) 
 &  \leq \sum_{u \in \widetilde{\mathcal{N}}_s} \Prob \left( u^T \hat{D}^* u > \frac{h(C_1, C_2, C_3, t)}{2} \right) \\
 & = \sum_{u \in \widetilde{\mathcal{N}}_s} \Prob \left( u^T \hat{D}^* u > h \left( \frac{C_1}{2}, \frac{C_2}{2}, C_3, t \right) \right)\\
 & \leq \left( \cfrac{9ep}{s} \right)^s e^{-t} = \delta \,,
 \end{split}\]
where the last equality holds when $t = s \log (9ep / s) + \log (1 / \delta)$.
\end{proof}

%%%%%%%%%%%%%%
%% Theorem 6 
%%%%%%%%%%%%%%%
\begin{theorem}[Signal under $L_0$ constraint] 
\label{thm:signal-l0}
Under assumptions (A1)-(A2), for any $\delta > 0$, there exist constants $C_1, C_2$, such that with probability at least $1 - \delta$,
\[ \lambda_1^R (\hat{D}) \geq \lambda_1^R (D) - C_1 \frac{\nu^2 \log(2 / \delta)}{n} -  \sqrt{C_2 \frac{\nu^4 \log (2 / \delta)}{n} }\,. \] 
\end{theorem}

\begin{proof}
Let $u_0 \in B_0(R) = \{ u: ||u||_2=1, ||u||_0 \leq R \}$ such that $\lambda_1^R (D) = u_0^T D u_0$. Then
\[ \lambda_1^R (\hat{D}) - \lambda_1^R (D) \geq u_0^T(  \hat{D}  - D )u_0 \,. \]
Therefore, it suffices to bound
\[ \left| u_0^T(  \hat{D}  - D )u_0 \right| \leq \left| u_0^T(  \hat{\Sigma}_1  - \Sigma_1 )u_0 \right| +
 \left| u_0^T(  \hat{\Sigma}_2  - \Sigma_2 )u_0 \right|\,.  \]
For any $\epsilon > 0$, note that $X_k^T u_0$ is sub-gaussian for $\forall k$, so by standard results (for example, Lemma 1 in \cite{ravikumar2011high}),
\[\begin{split}
 \Prob \left( \left| u_0^T(  \hat{\Sigma}_1  - \Sigma_1 )u_0 \right| > \epsilon \right)& = \Prob \left( \left| \frac{1}{n} \sum_{k=1}^n (X_k^T u_0)^2 - \Exp \left[ (X_1^T u_0)^2 \right] \right| > \epsilon \right)  \\
 & \leq 2 \exp \left\{ - \frac{n \epsilon^2}{ C_1' \nu^4 + C_2' \nu^2 \epsilon } \right\}
 \end{split}\]
 for some constants $C_1', C_2'$. The same arguments hold for $u_0^T(  \hat{\Sigma}_2  - \Sigma_2 )u_0$.
\end{proof}

Now we are able to state the proof for \Cref{thm:power-l0}.

\begin{proof}[Proof of \cref{thm:power-l0}]
By \Cref{thm:perm-l0} and \Cref{thm:signal-l0}, together with assumption (A3'), we know that for any $\delta > 0$, there exist some constants $C_1, C_2, C_3$ depending on $(\underline{c}, \bar{c}, \nu^2, \delta)$, such that with probability at least $1 - \delta$,
\[ \lambda_1^R (\hat{D}^*) \leq C_1 \sqrt{R \log (C_2 p) / n} \,, \ \
\lambda_1^R(\hat{D}) \geq \lambda_1^R(D) - C_3 \sqrt{1/n}\,.
\]
The same arguments hold for $-\hat{D}^*$ and $-\hat{D}$. Therefore, under assumption (A4') with some constant $C$ depending on $(\underline{c}, \bar{c}, \nu^2, \delta)$, we have
\[ \Prob_{H_1} \left( T_R(\hat{D}^*) > T_R(\hat{D}) \right) \leq \delta.  \]
The remaining statement follows by setting $\delta = \alpha/2$ and applying the Hoeffding's bound on the sample mean of Bernoulli random variables.
\end{proof}

\section{Lemmas}
%!TEX root = CovTest-sPCA-supplement.tex

In this section, we state and prove the lemmas that are used in \Cref{sec:proof-l1} and \Cref{sec:proof-l0}.
%%%%%%%%%%%%
%% lemma:bernstein
%%%%%%%%%%%%
\begin{lemma}[Bernstein inequality for sampling without replacement]
\label{lemma:bernstein}
Let $\mathcal{Z}= \{z_1, ..., z_N\}$ be a finite set containing $N$ real numbers, and $(z_1^*, ..., z_n^*)$ be $i.i.d.$ random variables that are drawn without replacement from $\mathcal{Z}$. Let
\[  \bar{z} = \max_{1 \leq i \leq N} |z_i|\,, \ 
 \mu_z = \frac{1}{N} \sum_{i=1}^N z_i\,, \  \sigma_z^2 = \frac{1}{N} \sum_{i=1}^N (z_i - \mu_z)^2\,, \]
then for any $\epsilon > 0$, 
\[ \Prob \left(  \left| \frac{1}{n}\sum_{i=1}^n z_i^* - \mu_z  \right| \geq \epsilon \right) \leq  2 \exp \left\{ - \frac{n \epsilon^2}{2 \sigma_z^2 + \frac{4}{3} \bar{z} \epsilon} \right\}\,. \]
As a consequence, for any $t > 0$,
\[ \Prob \left( \left| \frac{1}{n}\sum_{i=1}^n z_i^* - \mu_z \right|  > \frac{4 \bar{z}}{3} \frac{t}{n} + \sqrt{ 2 \sigma_z^2 \frac{t}{n}} \right) \leq  2 e^{-t}\,. \]
\end{lemma}

\begin{proof}
See Proposition 1.4 in \cite{bardenet2015concentration}.
\end{proof}

%%%%%%%%%%%%
%% lemma:sub-gaussian
%%%%%%%%%%%%
\begin{lemma}[Sub-gaussian tail bound]
\label{lemma:sub-gaussian}
Under assumptions (A1)-(A2), for $\forall \delta > 0$, there exist constants $C_1, C_2$ depending on $(\underline{c}, \bar{c}, \nu^2)$, such that if $(n, p)$ are sufficiently large, with probability at least $1 - \delta$,
\begin{enumerate}[(i)]
\item $|| \hat{\Sigma}_q - \Sigma_q ||_{\infty} \leq C_2 \sqrt{ \frac{ \log (C_1 p^2 / \delta)  }{N} }$ for $q = 1, 2$. As a consequence, 
\[|| \hat{D} - D||_\infty \leq 2C_2 \sqrt{ \frac{ \log (C_1 p^2 / \delta)  }{N} }\,. \]
\item $\overline{m}_z \leq C_2 \sqrt{ \frac{\log (C_1 p / \delta) }{N} }$. This together with (i) imply that
\[\overline{m}_z^{(2)} \leq 2\nu^{2} + 2C_2 \sqrt{ \frac{\log (C_1 p^2 / \delta)}{N} }\,.\]
\item $m_z \leq C_2 \sqrt{ \log (C_1 Np / \delta) }$.
\item 
%$\max_{i, j} \left| \frac{1}{n} \sum_{k=1}^n Z_{ki}^2 Z_{kj}^2 - \Exp(Z_{1i}^2)\Exp(Z_{1i}^2) \right| \leq  C_2 \sqrt{ \frac{\log (C_1 p^2 / \delta)}{n} }$. As a consequence, 
$ \overline{m}_z^{(4)} \leq C_2 \left[ 1 + \frac{\log (C_1 p^2 / \delta)}{N} \right] $.
\end{enumerate}
 \end{lemma}

\begin{proof} 
\begin{enumerate}[(i)]
\item See for example, Lemma 12 in \cite{yuan2010high}.

\item The first part is standard Hoeffding's bound on $\frac{1}{N} \sum_{k=1}^N Z_{ki}$, with a union bound over $1 \leq i \leq p$. The second part follows from 
\[ \overline{m}_z^{(2)} \leq \max \{ ||\hat{\Sigma}_1||_{\infty},   ||\hat{\Sigma}_2||_{\infty}\} + (\overline{m}_z)^2\,. \]

\item By Markov inequality, $\forall \epsilon,  t > 0$, 
\[\begin{split}
 \Prob \left( \max_{k, i} Z_{ki} > \epsilon \right) & 
 \leq e^{- t \epsilon}\Exp \left[ e^{t \max_{k, i} Z_{ki}}  \right] 
= e^{- t \epsilon} \Exp \left[  \max_{k, i} e^{t  Z_{ki} } \right]  \\
& \leq e^{- t \epsilon}  \sum_{k=1}^N \sum_{i=1}^p \Exp \left[ e^{t Z_{ki}} \right] \leq Np \cdot e^{- t \epsilon + \frac{ t^2 \nu^2 }{2}}\,.
 \end{split}\]
Finally, take $t = \frac{\epsilon}{\nu^2}$, and note that similar arguments hold for $- Z_{ki}$.

\item For any given $(i, j)$, let $W_k = Z_{ki}^2 Z_{kj}^2$, and define its cumulant generating function
\[  \Psi_k(\theta) = \log \Exp \left[ e^{\theta (W_k - \Exp(W_k))} \right]\,. \]
Note that $\Psi_1= \cdots = \Psi_n$ and $\Psi_{n+1} = \cdots = \Psi_{n+m}$. 
 By Markov inequality, for any $t, \theta > 0$, 
\begin{equation}
 \Prob \left( \left| \frac{1}{n} \sum_{k=1}^n W_k - \Exp(W_1) \right| > t \right) \leq 2 \exp \left\{ -n \theta t + n \Psi(\theta) \right\}\,,
 \label{eq:chernoff-a}  
 \end{equation}
where $\Psi(\theta) = \max\{ \Psi_1(\theta), \Psi_{n+1}(\theta) \}$ is an upper bound of the cumulant generating functions. 
Since $Z_{ki}, Z_{kj}$ are sub-gaussian, there exists a small constant $\theta_0 \neq 0$, such that 
$\Psi(\theta_0) < \infty$.
Plugging in $\theta_0$ to \cref{eq:chernoff-a}, we know that with probability at least $1 - \frac{\delta}{2} p^{-2}$, 
\[ \frac{1}{n} \sum_{k=1}^n W_k - \Exp(W_1)  \leq \frac{\log (4 p^2 / \delta) }{n \theta_0}  + \frac{\Psi(\theta_0)}{ \theta_0}\,. \]
The same arguments also hold for $\frac{1}{m} \sum_{k=(n+1)}^{n+m} W_k - \Exp(W_{n+1})$. Then the final result follows from a union bound over $(i, j)$ and the fact that $\Exp(W_k) \leq C \nu^4$ for some constant $C$.
\end{enumerate}
\end{proof}

\begin{lemma}
\label{lemma:sub-tail}
Under the same conditions as \Cref{thm:perm-l0}, let $ \widetilde{\mathcal{N}}_s \subseteq B_0(s) = \{ u \in \mathbb{R}^p: ||u||_2 = 1, ||u||_0 \leq s \}$ be a finite set such that $\left| \widetilde{\mathcal{N}}_s \right| < \infty$. Define
\[ \gamma_z = \max_{u \in \widetilde{\mathcal{N}}_s} \max_{1 \leq k \leq N} Z_k^T u\,, \
  \bar{\gamma}_z^{(4)} = \max_{u \in \widetilde{\mathcal{N}}_s} \frac{1}{N} \sum_{k=1}^N (Z_k^T u)^4\,.
 \]
Then for any $t > 0$, there exist constants $C_1, C_2, C_3$ depending on $\nu^2$, such that with probability at least $1 - e^{-t}$,
\[ \gamma_z \leq C_1 \sqrt{s } \sqrt{\log (C_3 Np) + t}\,, \  \bar{\gamma}_z^{(4)} \leq C_2 \left[ 1 +  \frac{t + \log \left| \widetilde{\mathcal{N}}_s \right|}{N}\right]\,. \] 
\end{lemma}

\begin{proof}
Let $m_z = ||Z||_{\infty} = \max_{1 \leq k \leq N, 1 \leq j \leq p} Z_{kj}$, and note that
\[ \gamma_z = \max_{u \in\widetilde{\mathcal{N}}_s} \max_{1 \leq k \leq N} Z_k^T u \leq m_z \cdot \max_{u \in \widetilde{\mathcal{N}}_s} \sum_{j=1}^{s} u_j \leq m_z \sqrt{s}\,. \]
Then the first result follows from \Cref{lemma:sub-gaussian}.

Next, for any $u \in  \widetilde{\mathcal{N}}_s$, denote $W_{u, k} = (Z_k^T u)^4$, with cumulant generating function
\[  \Psi_k(\theta) = \log \Exp \left[ e^{\theta \left[ W_{u, k} - \Exp(W_{u, k}) \right]} \right]\,. \]
Note that $\Psi_1= \cdots = \Psi_n$ and $\Psi_{n+1} = \cdots = \Psi_{n+m}$. 
Then for any $\epsilon, \theta > 0$, by Markov inequality, 
\begin{equation}
 \Prob \left( \frac{1}{N} \sum_{k=1}^N \left( W_{u, k} - \Exp(W_{u, k}) \right) > \epsilon \right) \leq \exp \left\{ - N \theta \epsilon + N \Psi(\theta) \right\}\,,
 \label{eq:chernoff}  
 \end{equation}
where $ \Psi(\theta) = \max\{\Psi_1(\theta), \Psi_{n+1}(\theta)  \}$. Note that $\{W_{u, k}\}_{k=1, ..., N}$ are sub-gaussian, so there exists a small constant $\theta_0 \neq 0$, such that
$\Psi(\theta_0) < \infty$.
Therefore, plugging $\theta_0$ into \cref{eq:chernoff}, we know that with probability at least $1 -  e^{-t} / |  \widetilde{\mathcal{N}}_s |$,
\[  \frac{1}{N} \sum_{k=1}^N \left( W_{u, k} - \Exp(W_{u, k}) \right) \leq \frac{t + \log \left| \widetilde{\mathcal{N}}_s \right|}{N \theta_0} + \frac{\Psi(\theta_0)}{\theta_0}\,.  \]
Finally, the desired result follows from a union bound over $u \in  \widetilde{\mathcal{N}}_s$ and the fact that $\Exp(W_{u, k}) \leq C \nu^4$ for some constant $C$.
\end{proof}

\end{document}